\documentclass{amsart} %!PN
\usepackage{color} 
\usepackage{amsfonts,color,pslatex}
\usepackage{amssymb,amsmath,latexsym}

\newcommand{\tr}{{\rm Tr}}
\newcommand{\gf}{{\rm GF}}
\newcommand{\Z}{\mathbb {Z}}
\newcommand{\N}{\mathbb {N}}
\newcommand{\m}{\mathbb {M}}

\newcommand{\C}{{\mathcal C}}
\newcommand{\seq}{{\check{s}}}
\newcommand{\ls}{{\mathbb L}}

\newtheorem{theorem}{Theorem}[section]
\newtheorem{lemma}[theorem]{Lemma}
\newtheorem{corollary}[theorem]{Corollary}

\newtheorem{example}[theorem]{Example}

\newtheorem{remark}[theorem]{Remark} 

\newtheorem{open}[theorem]{Open Problem}

\begin{document}

\title{Cyclic Codes from APN and Planar Functions}
\author{Cunsheng Ding} 
\address{Department of Computer Science and Engineering, 
The Hong Kong University of Science and Technology, Clear Water Bay, 
Kowloon, Hong Kong.
} 
\email{cding@ust.hk} 


\date{\today}

\begin{abstract} 
Cyclic codes are a subclass of linear codes and have applications in consumer electronics, 
data storage systems, and communication systems as they have efficient encoding and 
decoding algorithms. In this paper, almost perfect nonlinear functions and planar functions 
over finite fields are employed to construct a number of classes of cyclic codes. Lower bounds 
on the minimum weight of some classes of the cyclic codes are developed. The minimum 
weights of some other classes of the codes constructed in this paper are determined. The 
dimensions of the codes are flexible.  Many of the codes presented  in this paper are optimal 
or almost optimal in the sense that they meet some bound on linear codes. Ten open problems 
regarding cyclic codes from highly nonlinear functions are also presented. 
\end{abstract}

\maketitle

\section{Introduction}

Let $q$ be a power of a prime $p$. 
A linear $[n,k, d]$ code over $\gf(q)$ is a $k$-dimensional subspace of $\gf(q)^n$ 
with minimum (Hamming) nonzero weight $d$. 
A linear $[n,k]$ code $\C$ over the finite field $\gf(q)$ is called {\em cyclic} if 
$(c_0,c_1, \cdots, c_{n-1}) \in \C$ implies $(c_{n-1}, c_0, c_1, \cdots, c_{n-2}) 
\in \C$.  
Let $\gcd(n, q)=1$. By identifying any vector $(c_0,c_1, \cdots, c_{n-1}) \in \gf(q)^n$ 
with  
$$ 
c_0+c_1x+c_2x^2+ \cdots + c_{n-1}x^{n-1} \in \gf(q)[x]/(x^n-1), 
$$
any code $\C$ of length $n$ over $\gf(q)$ corresponds to a subset of $\gf(q)[x]/(x^n-1)$. 
The linear code $\C$ is cyclic if and only if the corresponding subset in $\gf(q)[x]/(x^n-1)$ 
is an ideal of the ring $\gf(q)[x]/(x^n-1)$. 
It is well known that every ideal of $\gf(q)[x]/(x^n-1)$ is principal. Let $\C=(g(x))$ be a 
cyclic code. Then $g(x)$ is called the {\em generator polynomial} and 
$h(x)=(x^n-1)/g(x)$ is referred to as the {\em parity-check} polynomial of 
$\C$. 

A vector $(c_0, c_1, \cdots, c_{n-1}) \in \gf(q)^n$ is said to be {\em even-like} 
if $\sum_{i=0}^{n-1} c_i =0$, and is {\em odd-like} otherwise. The minimum weight 
of the even-like codewords, respectively the odd-like codewords of a code is the 
minimum even-like weight, denoted by $d_{even}$, respectively the minimum 
odd-like weight of the code, denoted by $d_{odd}$.  The {\em even-like subcode} of a 
linear code consists of all the even-like codewords of this linear code. 

The error correcting capability of cyclic codes may not be as good as some other linear 
codes in general. However, cyclic codes have wide applications in storage and communication 
systems because they have efficient encoding and decoding algorithms 
\cite{Chie,Forn,Pran}. 
For example, Reed–-Solomon codes have found important applications from deep-space 
communication to consumer electronics. They are prominently used in consumer 
electronics such as CDs, DVDs, Blu-ray Discs, in data transmission technologies 
such as DSL \& WiMAX, in broadcast systems such as DVB and ATSC, and in computer 
applications such as RAID 6 systems.

Cyclic codes have been studied for decades and a lot of  progress has been made 
(see for example, \cite{IBM,Char,HPbook,LintW}).  The total number of cyclic codes 
over $\gf(q)$ and their constructions are closely related to cyclotomic cosets 
modulo $n$, and thus many areas of number theory. One way of   
constructing cyclic codes over $\gf(q)$ with length $n$ is  to use the generator polynomial 
\begin{eqnarray}\label{eqn-defseqcode}
\frac{x^n-1}{\gcd(S(x), x^n-1)}
\end{eqnarray}
where 
$$ 
S(x)=\sum_{i=0}^{n-1} s_i x^i  \in \gf(q)[x]   
$$
and $s^{\infty}=(s_i)_{i=0}^{\infty}$ is a sequence of period $n$ over $\gf(q)$. 
Throughout this paper, we call the cyclic code $\C_s$ with the generator polynomial 
of (\ref{eqn-defseqcode}) the {\em code defined by the sequence} $s^{\infty}$, 
and the sequence $s^{\infty}$ the {\em defining sequence} of the cyclic code $\C_s$. 

One basic question is whether good cyclic codes can be constructed with 
this approach. It will be demonstrated in this paper that the code $\C_s$ could 
be an optimal or almost optimal linear code if the sequence $s^\infty$ is properly 
designed.  

In this paper, almost perfect nonlinear (APN) functions and planar functions over 
$\gf(q^m)$ will be employed to construct a number of classes of both binary and 
nonbinary cyclic codes. Lower bounds on the minimum weight of some classes 
of the cyclic codes are developed. The minimum weights of some other classes of 
the codes constructed in this paper are determined. The dimensions of the codes 
of this paper are flexible.  Some of the codes obtained in this paper are optimal 
or almost optimal as they meet certain bounds on linear codes. Ten open problems 
regarding cyclic codes from highly nonlinear functions are also presented in this 
paper. 

Our first motivation 
of this study is that the codes constructed in this paper are often optimal. Our 
second motivation is the simplicity of the constructions of the cyclic codes that 
may lead to efficient encoding and decoding algorithms.

\section{Preliminaries} 

In this section, we present basic notations and results of almost perfect nonlinear    
and planar functions, $q$-cyclotomic cosets, and sequences that will be employed 
in subsequent sections.     

\subsection{Some notations fixed throughout this paper}\label{sec-notations} 

Throughout this paper, we adopt the following notations unless otherwise stated: 
\begin{itemize} 
\item $p$ is a prime. 
\item $q$ is a positive power of $p$. 
\item $m$ is a positive integer. 
\item $r=q^m$. 
\item $n=q^m-1$.  
\item $\Z_n=\{0,1,2,\cdots, n-1\}$ associated with the integer addition modulo $n$ and  
           integer multiplication modulo $n$ operations. 
\item $\alpha$ is a generator of $\gf(r)^*$. 
\item $m_a(x)$ is the minimal polynomial of $a \in \gf(r)$ over $\gf(q)$. 
\item $\N_p(x)$ is a function defined by $\N_p(i) =0$ if $i \equiv 0 \pmod{p}$ and $\N_p(i) =1$ otherwise, where $i$ 
          is any nonnegative integer.    
\item $\tr(x)$ is the trace function from $\gf(r)$ to $\gf(q)$.           
\end{itemize}

\subsection{The $q$-cyclotomic cosets modulo $q^m-1$}\label{sec-cpsets}

The $q$-cyclotomic coset containing $j$ modulo $n$ is defined by 
$$ 
C_j=\{j, qj, q^2j, \cdots, q^{\ell_j-1}j\} \subset \Z_n 
$$
where $\ell_j$ is the smallest positive integer such that $q^{\ell_j-1}j \equiv j \pmod{n}$, 
and is called the size of $C_j$. It is known that $\ell_j$ divides $n$. The smallest integer 
in $C_j$ is called the {\em coset leader} of $C_j$. Let $\Gamma$ denote the set of all 
coset leaders. By definition, we have 
$$ 
\bigcup_{j \in \Gamma} C_j =\Z_n.  
$$

\subsection{The linear span and minimal polynomial of sequences}

Let $s^L=s_0s_1\cdots s_{L-1}$ be a sequence over $\gf(q)$. The {\em linear 
span} (also called {\em linear complexity}) of $s^L$ is defined to be the smallest positive 
integer $\ell$ such that there are constants $c_0=1, c_1, \cdots, c_\ell \in \gf(q)$ 
satisfying 
\begin{eqnarray*} 
-c_0s_i=c_1s_{i-1}+c_2s_{i-2}+\cdots +c_ls_{i-\ell} \mbox{ for all } \ell \leq i<L. 
\end{eqnarray*} 
In engineering terms, such a polynomial $c(x)=c_0+c_1x+\cdots +c_lx^l$ 
is called the {\em feedback  polynomial} of a shortest linear feedback 
shift register 
(LFSR) that generates $s^L$. Such an integer always exists for finite sequences  $s^L$. When 
$L$ is $\infty$, a sequence $s^{\infty}$ is called a semi-infinite 
sequence. If there is no such an integer for a semi-infinite sequence 
$s^{\infty}$, its linear span is defined to be $\infty$. The linear 
span of the zero sequence is defined to be zero. 
For ultimately periodic semi-infinite sequences such an $\ell$ always 
exists. 

Let $s^{\infty}$ be a sequence of period $L$ over $\gf(q)$. 
Any feedback polynomial of $s^{\infty}$ is called a {\em characteristic 
polynomial}. The characteristic polynomial with the smallest degree is 
called the {\em minimal polynomial} of the periodic sequence $s^{\infty}$. 
Since we require that the constant term of any characteristic polynomial 
be 1, the minimal polynomial of any periodic sequence $s^{\infty}$ must 
be unique. In addition, any characteristic polynomial must be a multiple 
of the minimal polynomial.    

For periodic sequences, there are a few ways to determine their linear 
span and minimal polynomials. One of them is given in the following 
lemma \cite{LN97}. 

\begin{lemma}\label{lem-ls1} 
Let $s^{\infty}$ be a sequence of period $L$ over $\gf(q)$. 
Define   
\begin{eqnarray*}
S^{L}(x)=s_{0}+s_{1}x+\cdots +s_{L-1}x^{L-1} \in \gf(q)[x]. 
\end{eqnarray*}
Then the minimal polynomial $\m_s(x)$ of $s^{\infty}$ is given by 
      \begin{eqnarray}\label{eqn-base1}  
      \frac{x^{L}-1}{\gcd(x^{L}-1, S^{L}(x))} 
      \end{eqnarray}  
and the linear span $\ls_s$ of $s^{\infty}$ is given by 
      \begin{eqnarray}\label{eqn-base2} 
       L-\deg(\gcd(x^{L}-1, S^{L}(x))). 
      \end{eqnarray}  
\end{lemma} 

The other one is given in the following lemma \cite{Antweiler} 

\begin{lemma} \label{lem-ls2} 
Any sequence $s^{\infty}$ over $\gf(q)$ of period $q^m-1$ has a unique expansion of the form  
\begin{equation*}
s_t=\sum_{i=0}^{q^m-2}c_{i}\alpha^{it}, \mbox{ for all } t\ge 0,
\end{equation*}
where $\alpha$ is a generator of $\gf(q^m)^*$ and $c_i \in \gf(q^m)$.
Let the index set $I=\{i \left.\right| c_i\neq 0\}$, then the minimal polynomial $\m_s(x)$ of $s^{\infty}$ is 
\begin{equation*}
\m_s(x)=\prod_{i\in I}(1-\alpha^i x),
\end{equation*}
and the linear span of $s^{\infty}$ is $|I|$.
\end{lemma}

It should be noticed that in some references the reciprocal of $\m_s(x)$ is called the minimal polynomial 
of the sequence $s^\infty$. So Lemma \ref{lem-ls2} is a modified version of the original one in \cite{Antweiler}.

\subsection{Perfect and almost perfect nonlinear functions on $\gf(r)$}\label{sec-APNPN} 

A function $f: \gf(r) \to \gf(r)$ is called {\em almost perfect 
nonlinear (APN)} if 
$$ 
\max_{a \in \gf(r)^*} \max_{b \in \gf(r)} |\{x \in \gf(r): f(x+a)-f(x)=b\}| =2,  
$$ 
and is referred to as 
{\em perfect 
nonlinear or planar} if 
$$ 
\max_{a \in \gf(r)^*} \max_{b \in \gf(r)} |\{x \in \gf(r): f(x+a)-f(x)=b\}| =1.   
$$ 

There is no perfect nonlinear (planar) function on $\gf(2^m)$. The following is a summary of 
known APN monomials $x^e$ over $\gf(2^m)$:  
\begin{itemize} 
\item $e=2^m-2$, $m$ odd (\cite{BD,Nybe}). 
\item $e=2^h+1$ with $\gcd(h, m)=1$, where $1 \leq h \leq (m-1)/2$ 
      if $m$ is odd and $1 \leq h \leq (m-2)/2$ if $m$ is even 
      (\cite{Gold}). 
\item $e=2^{2h}-2^h+1$ with $\gcd(h, m)=1$, where $1 \leq h \leq (m-1)/2$ 
      if $m$ is odd and $1 \leq h \leq (m-2)/2$ if $m$ is even 
      (\cite{Kasa}). 
\item $e=2^{(m-1)/2}+3$, where $m$ is odd (\cite{Dobb99,HX}).  
\item $e=2^{(m-1)/2}+2^{(m-1)/4}-1$, where $m \equiv 1 \pmod{4}$ 
      (\cite{Dobb992,HX}).
\item $e=2^{(m-1)/2}+2^{(3m-1)/4}-1$, where $m \equiv 3 \pmod{4}$ 
      (\cite{Dobb992,HX}).
\item $e=2^{4i}+2^{3i}+2^{2i}+2^i-1$, where $m=5i$ (\cite{Dobb992}).        
\end{itemize}   
A number of other types of APN functions $f(x)$ on $\gf(2^m)$ were discovered in \cite{BC,BCL2,BCL1,BCP}. 

The following is a summary of known APN monomials $x^e$ over $\gf(p)^m$ where $p$ is odd:  
\begin{itemize} 
\item $e=3$, $p >3$ (\cite{HRS}). 
\item $e=p^m-2$, $p>2$ and $p \equiv 2 \pmod{3}$ (\cite{HRS}).  
\item $e=\frac{p^m-3}{2}$, $p \equiv 3, 7 \pmod{20}$, $p^m>7$, $p^m \ne 27$ and $m$ is odd (\cite{HRS}).  
\item $e=\frac{p^m+1}{4} + \frac{p^m-1}{2}$, $p^m \equiv 3 \pmod{8}$ (\cite{HRS}).  
\item $e=\frac{p^m+1}{4}$, $p^m \equiv 7 \pmod{8}$ (\cite{HRS}).  
\item $e=\frac{2p^m-1}{3}$, $p^m \equiv 2 \pmod{3}$ (\cite{HRS}).  
\item $e=p^m-3$, $p=3$ and $m$ is odd. 
\item $e=p^l+2$, $p^l \equiv 1 \pmod{3}$ and $m=2l$ (\cite{HRS}). 
\item $e=\frac{5^h+1}{2}$, $p=5$ and $\gcd(2m,h)=1$. 
\item $e=\left(3^{(m+1)/4}-1\right)\left(3^{(m+1)/2}+1\right)$, $m \equiv 3 \pmod{4}$ and $p=3$ \cite{Zha}. 
\item Let $p=3$, and \begin{eqnarray*} 
e=\left\{ \begin{array}{ll} 
               \frac{3^{(m+1)/2}-1}{2} & \mbox{if } m \equiv 3 \pmod{4} \\
               \frac{3^{(m+1)/2}-1}{2} + \frac{3^m-1}{2} & \mbox{if } m \equiv 1 \pmod{4}.                 
               \end{array} 
\right.                
\end{eqnarray*}   
\item Let $p=3$, and \begin{eqnarray*} 
e=\left\{ \begin{array}{ll} 
               \frac{3^{m+1}-1}{8} & \mbox{if } m \equiv 3 \pmod{4} \\
               \frac{3^{m+1}-1}{8} + \frac{3^m-1}{2} & \mbox{if } m \equiv 1 \pmod{4}.                 
               \end{array} 
\right.                
\end{eqnarray*}   
\item $e=\frac{5^m-1}{4} + \frac{5^{(m+1)/2}-1}{2}$, $p=5$ and $n$ is odd \cite{Zha}. 
\end{itemize}   

The following is a list of some known palnar functions over $p^m$, $p$ odd: 
\begin{itemize} 
\item $f(x)=x^2$. 
\item $f(x)=x^{p^h+1}$, where $m/\gcd(m,h)$ is odd (\cite{DO}).  
\item $f(x)=x^{(3^h+1)/2}$, where $p=3$ and $\gcd(m,h)=1$ (\cite{CM}). 
\item $f(x)=x^{10}-ux^6-u^2x^2$, where $p=3$, $u \in \gf(p^m)$, $m$ is odd (\cite{CM,DY06}).  
\end{itemize} 
Recently, more planar functions were discovered in \cite{ZKW,ZW,Zha}. 

In the sequel, some of the APN functions and all the planar functions on $\gf(q^m)$ above will be 
employed to construct cyclic codes over $\gf(q)$. 

\subsection{Codes defined by highly nonlinear functions}\label{sec-sequence} 

Given any function $f(x)$ on $\gf(r)$, we define its associated sequence 
$s^\infty$ by 
\begin{eqnarray}\label{eqn-sequence}
s_i=\tr(f(\alpha^i+1)) 
\end{eqnarray}
for all $i \ge 0$, where $\alpha$ is a generator of $\gf(r)^*$ and $\tr(x)$ denotes 
the trace function from $\gf(r)$ to $\gf(q)$.

The objective of this paper is to consider the codes $\C_s$ defined by planar 
functions and APN functions over $\gf(q^m)$. We need to treat the cases $q=2$ and $q$ being odd  
separately as planar functions on $\gf(2^m)$ do not exist, while both APN and planar functions 
on $\gf(q^m)$ exist when $q$ is odd.   

Highly nonlinear (i.e., almost perfect nonlinear, perfect nonlinear and bent) functions were 
employed to construct linear codes with good parameters in \cite{CCD,CCZ,CDY05}. The 
approach to the constructions of cyclic codes with APN and planar functions employed in this 
paper is quite different.

\section{Binary cyclic codes from APN functions on $\gf(2^m)$}\label{sec-binary} 

Planar functions on $\gf(2^m)$ do not exist. In this section, we treat binary cyclic codes 
derived from APN functions on $\gf(2^m)$, and fix $q$ to be 2 throughout this section. 

\subsection{Binary cyclic codes from the inverse APN function}\label{sec-Inverse} 

In this subsection we study the code $\C_s$ defined by the inverse APN function on $\gf(2^m)$. 
To this end, we need to prove the following lemma.  

Let $\rho_i$ denote the total number of even integers in the $2$-cyclotomic coset $C_i$. 
We then define 
\begin{eqnarray}\label{eqn-defnu}
\nu_i=\frac{m\rho_i}{\ell_i} \bmod{2} 
\end{eqnarray} 
for each $i \in \Gamma$, where $\ell_i=|C_i|$.

\begin{lemma}\label{lem-Inverse} 
Let $s^{\infty}$ be the sequence of (\ref{eqn-sequence}), where $f(x)=x^{2^m-2}$. Then the linear span $\ls_s$
of $s^{\infty}$ is equal to $(n+1)/2$ and the minimal polynomial $\m_s(x)$ of  $s^{\infty}$ is given by 
\begin{equation}\label{eqn-gpInverse}
\m_s(x)=\prod_{j \in \Gamma, \nu_j=1} m_{\alpha^{-j}}(x)
\end{equation} 
where $m_{\alpha^j}(x)$ is the minimal polynomial of $\alpha^j$ over $\gf(2)$. 
\end{lemma} 

\begin{proof} 
The linear span of this sequence was already determined in \cite{wenpeiding}. Below we prove 
only the conclusion on the minimal polynomial of this sequence. 

It was proved in \cite{wenpeiding} that 
\begin{eqnarray}\label{eqn-case2}
s_t &=&\tr\left(\sum_{i=0}^{2^{m-1}-1} \alpha^{2it}\right)\nonumber \\
&=& \sum_{j \in \Gamma}\sum_{u=0}^{\ell_j-1}\alpha^{jt2^{u}}\left(\sum_{i=0}^{m-1}f_{j,i} \right), 
\end{eqnarray}
where 
\begin{equation}\label{eqn-con2}
f_{j,i}=\left \{ \begin{array}{ll}
1 & \mbox{if }(j\, 2^{m-i}\bmod{n}) \bmod{2}=0 \\
0 & \mbox{otherwise.}
 \end{array}\right.
\end{equation}

It then follows from (\ref{eqn-defnu}), (\ref{eqn-con2}) and (\ref{eqn-case2}) that 
\begin{eqnarray}\label{eqn-feb22}
s_t=\sum_{j \in \Gamma} \nu_j \left( \sum_{i \in C_j} (\alpha^t)^i \right).   
\end{eqnarray} 
The desired conclusion on the minimal polynomial $\m_s(x)$ then follows from Lemma \ref{lem-ls2} 
and (\ref{eqn-feb22}). 
\end{proof}

The following theorem provides information on the code $\C_{s}$ and its dual.    

\begin{theorem} 
The binary code $\C_{s}$ defined by the sequence of Lemma \ref{lem-Inverse} has parameters 
$[2^m-1, 2^{m-1}-1, d]$ and generator polynomial $\m_s(x)$ of (\ref{eqn-gpInverse}). 

If $m$ is odd, the minimum distance $d$ of $\C_{s}$ is at least $d_1$, where $d_1$ is the smallest 
positive even integer with $d_1^2-d_1+1 \ge n$, and the dual code $\C_{s}^\perp$ has parameters 
$[2^m-1, 2^{m-1}, d^\perp]$ where $d^\perp$ satisfies that  $(d^\perp)^2-d^\perp+1 \ge n$.    
\end{theorem} 

\begin{proof} 
The dimensions of $\C_{s}$ and its dual follow from Lemma \ref{lem-Inverse} and the definitions of the 
codes $\C_s$ and $\C_{s}^\perp$.  

If $m$ is odd, by Lemma \ref{lem-Inverse}, $m_1(x)=x-1$ is a divisor of the generator polynomial 
$\m_s(x)$ of  $\C_{s}$. Hence, all the codewords in $\C_{s}$ have even Hamming weights, i.e., 
$\C_s$ is an even-weight code. By definition, we have 
\begin{eqnarray}\label{eqn-march22}
\rho_i + \rho_{n-i} \equiv 1 \pmod{2} 
\end{eqnarray} 
for every $i \in \{1, 2, \cdots, n-1\}$. Note that $m$ is odd. It then follows from (\ref{eqn-march22}) 
that 
\begin{eqnarray*} 
\nu_i + \nu_{n-i} \equiv 1 \pmod{2}  
\end{eqnarray*}   
for every $i \in \{1, 2, \cdots, n-1\}$. 
Hence, one and only one of $m_{\alpha^{i}}(x)$ and $m_{\alpha^{-i}}(x)$ is a divisor of  $\m_s(x)/(x-1)$ 
for every $i \in \{1, 2, \cdots, n-1\}$.  
Let $\overline{\C_{s}}$ denote the cyclic code with generator polynomial $\m_s(x)/(x-1)$. Then  
$\overline{\C_{s}}$ contains $\C_{s}$ as its even-weight 
subcode. Using a similar approach to the proof of the square-root bound on the minimum weight for 
the quadratic residue 
codes, one proves the desired conclusions on the minimum weights of $\overline{\C_s}$ and $\C_{s}^\perp$.  
\end{proof}

\begin{example} 
Let $m=3$ and $\alpha$ be a generator of $\gf(2^m)^*$ with $\alpha^3 + \alpha + 1=0$. In this case, 
the three 2-cyclotomic cosets are 
$$ 
C_0=\{0\}, \ C_1=\{1,2,4\}, \ C_3=\{3,6,5\}. 
$$
The generator polynomial of the code $\C_s$ is 
$$ 
\m_s(x)=m_{\alpha^0}(x) m_{\alpha^{-3}}(x)=(x+1)(x^3 + x + 1)=x^4 + x^3 + x^2 + 1   
$$
and $\C_s$ is a $[7, 3, 4]$ binary cyclic code. Its dual is a $[7,4,3]$ cyclic code. Both codes are optimal. 
\end{example}

\begin{example} 
Let $m=5$ and $\alpha$ be a generator of $\gf(2^m)^*$ with $\alpha^5 + \alpha^2 + 1=0$. Then 
the generator polynomial of the code $\C_s$ is 
\begin{eqnarray*} 
\m_s(x) &=& m_{\alpha^0}(x) m_{\alpha^{-3}}(x)  m_{\alpha^{-5}}(x) m_{\alpha^{-15}}(x) \\
&=& (x + 1)(x^5 + x^3 + x^2 + x + 1) \times \\
& & (x^5 + x^4 + x^3 + x + 1) (x^5 + x^2 + 1) \\
&=& x^{16} + x^{14} + x^{13} + x^{10} + x^9 + x^8 + \\ 
& & x^7 + x^6 + x^5 + x^2 + x + 1.      
\end{eqnarray*}
and $\C_s$ is a $[31, 15, 8]$ binary cyclic code. Its dual is an $[31,16,7]$ cyclic code. Both codes are optimal. 

\end{example} 

When $m$ is odd, the code $\C_s$ has the square-root bound, although $n$ could be a composite number. 
In addition, the examples above show that the actual minimum weight could be much larger than the lower 
bound on the minimum weight.   

When $m$ is even, $x^{-1}$ is not APN. In this case, the code $\C_s$ may not have a good minimum distance. 
For example, if $m=4$, the code $\C_s$ has parameters $[15, 7, 3]$. So we are not interested in the case 
that $m$ is even.

\subsection{Binary cyclic codes from the Gold APN function}\label{sec-Gold} 

In this subsection we study the code $\C_s$ defined by the Gold APN function. To this end, we need to 
prove the following lemma.  

\begin{lemma}\label{lem-Gold} 
Let $m$ be odd. 
Let $s^{\infty}$ be the sequence of (\ref{eqn-sequence}), where $f(x)=x^{2^h+1}$, $\gcd(h,m)=1$. 
Then the linear span $\ls_s$ of $s^{\infty}$ is equal to $m+1$ and the minimal polynomial $\m_s(x)$ 
of  $s^{\infty}$ is given by 
\begin{equation}\label{eqn-Gold}
\m_s(x)= (x-1) m_{\alpha^{-(2^h+1)}}(x)
\end{equation} 
where $m_{\alpha^{-(2^h+1)}}(x)$ is the minimal polynomial of $\alpha^{-(2^h+1)}$ over $\gf(2)$. 
\end{lemma} 

\begin{proof} 
It is easily seen that 
\begin{eqnarray}\label{eqn-Gold2}
s_t = 1+\tr\left(\alpha^{t(2^h+1)}\right) 
= 1+ \sum_{j=0}^{m-1} (\alpha^t)^{(2^h+1)2^j}. 
\end{eqnarray}

By assumption, $\gcd(h, m)=1$. We have then $\gcd(2^h-1, 2^m-1)=1$. It then follows that 
$$ 
\gcd(2^h+1, 2^m-1)=\gcd(2^{2h}-1, 2^m-1)=2^{\gcd(2h, m)}-1=1. 
$$
Therefore, the size of the 2-cyclotomic coset contating $2^h+1$ is $m$. 
The desired conclusions on the linear span and the minimal polynomial $\m_s(x)$ then follow from Lemma \ref{lem-ls2} 
and (\ref{eqn-Gold2}). 
\end{proof}

The following theorem provides information on the code $\C_{s}$.    

\begin{theorem} 
Let $m$ be odd. 
The binary code $\C_{s}$ defined by the sequence of Lemma \ref{lem-Gold} has parameters 
$[2^m-1, 2^{m}-2-m, d]$ and generator polynomial $\m_s(x)$ of (\ref{eqn-Gold}), where 
$d \ge 4$.  
\end{theorem} 

\begin{proof} 
The dimension of $\C_{s}$ follows from Lemma \ref{lem-Gold} and the definition of the 
code $\C_s$.  We need to prove the conclusion on the minimum distance $d$ of $\C_{s}$. 
To this end, let $\overline{\C_{s}}$ denote the cyclic code with generator polynomial $m_{\alpha^{-(2^h+1)}}(x)$. 
Then $\C_{s}$ is the even-weight subcode of   $\overline{\C_{s}}$. Since $m_{\alpha^{-(2^h+1)}}(x)$ is 
a primitive polynomial with order $2^m-1$, it does not divide $1+x^j$ for any $j$ with $2 \le j \le n-1$. 
This means that the minimum weight $\bar{d}$ of $\overline{\C_{s}}$ is at least 3. Since $d$ is even, 
we have then $d \ge 4$. 
\end{proof}

\begin{example} 
Let $(m,h)=(3,1)$ and $\alpha$ be a generator of $\gf(2^m)^*$ with $\alpha^3 + \alpha + 1=0$. Then 
$\C_s$ is a $[7, 3, 4]$ binary code with generator polynomial  
$$ 
\m_s(x)=x^4 + x^3 + x^2 + 1.    
$$
Its dual is a $[7,4,3]$ cyclic code. Both codes are optimal. 
\end{example}

\begin{example} 
Let $(m, h)=(5,1)$ and $\alpha$ be a generator of $\gf(2^m)^*$ with $\alpha^5 + \alpha^2 + 1=0$. Then 
the generator polynomial of the code $\C_s$ is 
\begin{eqnarray*} 
\m_s(x)=x^6 + x^5 + x^4 + 1
\end{eqnarray*}
and $\C_s$ is a $[31, 25, 4]$ binary cyclic code. Its dual is a $[31,6,15]$ cyclic code. Both codes are optimal. 

\end{example} 

\begin{example} 
Let $(m, h)=(7,2)$ and $\alpha$ be a generator of $\gf(2^m)^*$ with $\alpha^7 + \alpha + 1=0$. Then 
the generator polynomial of the code $\C_s$ is 
\begin{eqnarray*} 
\m_s(x)=x^8 + x^4 + x + 1
\end{eqnarray*}
and $\C_s$ is a $[127, 119, 4]$ binary cyclic code. Its dual is a $[127,8,63]$ cyclic code. Both codes are optimal. 
\end{example} 

When $m$ is odd, the code $\C_s$ is optimal by the sphere packing bound and may be equivalent to the even-weight 
subcode of the Hamming code. When $m$ is even, one can prove 
that $\C_s$ is a $[2^m-1, 2^m-1-m, 2]$ code and is almost optimal. We are not interested in this case as $d=2$.

\subsection{Binary cyclic codes from the Welch APN function}\label{sec-Welch} 

In this subsection we study the code $\C_s$ defined by the Welch APN function. Before doing this, we need to 
prove the following lemma.  

\begin{lemma}\label{lem-Welch} 
Let $m =2t+1 \ge 7$. 
Let $s^{\infty}$ be the sequence of (\ref{eqn-sequence}), where $f(x)=x^{2^{t}+3}$. 
Then the linear span $\ls_s$ of $s^{\infty}$ is equal to $5m+1$ and the minimal polynomial $\m_s(x)$ 
of  $s^{\infty}$ is given by 
\begin{eqnarray}\label{eqn-Welch}
\lefteqn{\m_s(x)=} \nonumber \\
& (x-1) m_{\alpha^{-1}}(x) m_{\alpha^{-3}}(x) m_{\alpha^{-(2^t+1)}}(x)m_{\alpha^{-(2^t+2)}}(x) m_{\alpha^{-(2^t+3)}}(x) 
    \nonumber \\ 
\end{eqnarray} 
where $m_{\alpha^{-j}}(x)$ is the minimal polynomial of $\alpha^{-j}$ over $\gf(2)$. 
\end{lemma} 

\begin{proof} 
By definition, we have  
\begin{eqnarray}\label{eqn-Welch2}
s_i &=& \tr\left((\alpha^i+1)^{2^t+2+1}\right) \nonumber \\ 
&=& \tr\left( (\alpha^i)^{2^t+3} + (\alpha^i)^{2^t+2} + (\alpha^i)^{2^t+1} +(\alpha^i)^{3} +\alpha^i +1    \right) \nonumber \\ 
&=& \sum_{j=0}^{m-1} (\alpha^i)^{(2^t+3)2^j} + \sum_{j=0}^{m-1}  (\alpha^i)^{(2^t+2)2^j} + \sum_{j=0}^{m-1}  (\alpha^i)^{(2^t+1)2^j} + 
        \nonumber \\
& &  \sum_{j=0}^{m-1}   (\alpha^i)^{3\times2^j} +  \sum_{j=0}^{m-1}   (\alpha^i)^{2^j} + 1. 
\end{eqnarray}

We will prove that the following $2$-cyclotomic cosets are pairwise disjoint: 
\begin{eqnarray}\label{eqn-fivecosets}
C_1, \ C_3, \ C_{2^t+1}, \ C_{2^t+2}, \ C_{2^t+3}.  
\end{eqnarray} 
Note that two cyclotomic cosets are either identical or disjoint. 
Since $C_1$ contains only powers of $2$, it cannot be identical with any of the remaining four 2-cyclotomic cosets. 

We now prove that $C_3 \cap C_{2^t+j}=\emptyset$ for all $j \in \{1,2,3\}$. Define 
$$ 
\Delta(h,j)=3\times 2^h-2^t-j 
$$ 
where $0 \le h \le m-1$.  When $h=m-1=2t$, we have 
$$ 
\Delta(h,j)=2^m-1+ 2^{m-1}-2^t-(j-1).  
$$ 
Note that $0 < 2^{m-1}-2^t-(j-1)<n$. It follows that $n$ does not divide $\Delta(h,j)$ in this case. 
When $h \le m-2=2t-1$, it is easily checked that $\Delta(h,j)\ne 0$ for any $h$ as $t \ge 3$. In this case 
we have 
$$ 
-n  < -(2^t+j) \le \Delta(h,j) =2^{m-1}+2^{m-2}-2^t -j < n. 
$$ 
Hence  $n$ does not divide $\Delta(h,j)$ for all $h \le m-2=2t-1$. It then follows that 
$C_3 \cap C_{2^t+j}=\emptyset$ for all $j \in \{1,2,3\}$.  

One can similarly prove that the three cyclotomic cosets $C_{2^t+1}, C_{2^t+2}, C_{2^t+3}$ are pairwise 
disjoint. We omit the details here.   
 
We now prove that all the five cyclotomic cosets of (\ref{eqn-fivecosets}) have size $m$. Clearly 
$\ell_1=|C_1|=\ell_{-1}=m$. Note that 
$$ 
\gcd(2^{2t+1}-1, 3)=\gcd(2^{2t+1}-1, 2^2-1)=2^{\gcd(2t+1,2)}-1=1. 
$$
We have $\ell_3=|C_3|=\ell_{-3}=m$. 

Since $\gcd(2^{2t+1}-1, 2^t-1)=1$, we have 
\begin{eqnarray*} 
\gcd(2^{2t+1}-1, 2^t+1)
&=& \gcd(2^{2t+1}-1, 2^{2t}-1) \\
&=& 2^{\gcd(2t+1, 2t)}-1 \\
&=& 1. 
\end{eqnarray*} 
Hence $\ell_{2^t+1}=|C_{2^t+1}|=\ell_{-(2^t+1)}=m$.

We now compute 
$$
\gcd:=\gcd(2^{2t+1}-1, 2^t+2)=\gcd(2^{2t+1}-1, 2^{t-1}+1).
$$ 
Note that 
$$ 
2^{2t+1}-1=2^{t+2}(2^{t-1}+1) -(2^{t+2}+1). 
$$
We have $\gcd=\gcd(2^{t+2}+1, 2^{t-1}+1)$. Since   
$$ 
2^{t+2}+1=2^3(2^{t-1}+1)-(2^3-1), 
$$
we obtain that $\gcd=\gcd(2^{t-1}+1, 2^3-1)$. 
Let $t_1=\lfloor (t-1)/3\rfloor$. Using the Euclidean division recursively, one gets 
\begin{eqnarray*}
\gcd &=& \gcd(2^{t-1-3t_1}+1, 2^3-1) \\
&=& \left\{ \begin{array}{ll} 
       \gcd(2^0+1, 2^3-1)=1 & \mbox{if } t-1 \equiv 0 \pmod{3} \\
       \gcd(2^1+1, 2^3-1)=1 & \mbox{if } t-1 \equiv 1 \pmod{3} \\
       \gcd(2^2+1, 2^3-1)=1 & \mbox{if } t-1 \equiv 2 \pmod{3}.               
       \end{array} 
       \right.        
\end{eqnarray*}
Therefore $\ell_{2^t+2}=|C_{2^t+2}|=\ell_{-(2^t+2)}=m$.

We now prove that  
$
\gcd:=\gcd(2^{2t+1}-1, 2^t+3)=1.
$
The conclusion is true for all $1 \le t \le 4$. So we consider only the case that $t \ge 5$.  

Note that 
$$ 
2^{2t+1}-1=(2^{t+1}-6)(2^{t}+3) + 17. 
$$
We have $\gcd=\gcd(2^{t}+3, 17)$. Since   
$$ 
2^{t}+3=2^{t-4}(2^{4}+1)-(2^{t-4}-3), 
$$
we obtain that $\gcd=\gcd(2^{t-4}-3, 2^3-1)$. 
Let $t_1=\lfloor t/4\rfloor$. Using the Euclidean division recursively, one gets 
\begin{eqnarray*}
\gcd &=& \gcd(2^{t-1-4t_1}+3 \times (-1)^{t_1}, 2^3-1) \\
&=& \left\{ \begin{array}{ll} 
       \gcd(2^0+(-1)^{t_1}3, 17)=1 & \mbox{if } t \equiv 0 \pmod{4}, \\
       \gcd(2^1+(-1)^{t_1}3, 17)=1 & \mbox{if } t \equiv 1 \pmod{4}, \\  
       \gcd(2^2+(-1)^{t_1}3, 17)=1 & \mbox{if } t \equiv 2 \pmod{4}, \\
       \gcd(2^3+(-1)^{t_1}3, 17)=1 & \mbox{if } t \equiv 3 \pmod{4}.          
       \end{array} 
       \right.        
\end{eqnarray*}
Therefore $\ell_{2^t+3}=|C_{2^t+3}|=\ell_{-(2^t+3)}=m$.

The desired conclusions on the linear span and the minimal polynomial $\m_s(x)$ then follow from Lemma \ref{lem-ls2}, 
(\ref{eqn-Welch2}) and the conclusions on the five cyclotomic cosets and their sizes. 
\end{proof}

The following theorem provides information on the code $\C_{s}$.    

\begin{theorem}\label{thm-Welch} 
Let $m \ge 7$ be odd. 
The binary code $\C_{s}$ defined by the sequence of Lemma \ref{lem-Welch} has parameters 
$[2^m-1, 2^{m}-2-5m, d]$ and  generator polynomial $\m_s(x)$ of (\ref{eqn-Welch}), where $d \ge 6$.  
\end{theorem} 

\begin{proof} 
The dimension of $\C_{s}$ follows from Lemma \ref{lem-Welch} and the definition of the 
code $\C_s$.  We need to prove the conclusion on the minimum distance $d$ of $\C_{s}$. 
To this end, let $\overline{\C_{s}}$ denote the cyclic code with generator polynomial $\m_s(x)/(x-1)$. 
Then $\C_{s}$ is the even-weight subcode of   $\overline{\C_{s}}$. 

Note that 
$$ 
C_{1} \cup C_{2^t+1} \cup  C_{2^t+2} \cup  C_{2^t+2}  \supset \{2^t, 2^t+1, 2^t+2, 2^t+3\}. 
$$ 
By the BCH bound, the cyclic code generated by the reciprocal of $\m_s(x)/(x-1)$ has minimum weight 
at least 5. So does $\overline{\C_{s}}$. Note that $d$ is even. It then follows that $d \ge 6$.  
\end{proof}

\begin{example} 
Let $m=3$ and $\alpha$ be a generator of $\gf(2^m)^*$ with $\alpha^3 + \alpha + 1=0$. Then 
$\C_s$ is a $[7, 3, 4]$ binary code with generator polynomial  
$$ 
\m_s(x)=x^4 + x^3 + x^2 + 1.    
$$
Its dual is a $[7,4,3]$ cyclic code. Both codes are optimal. 
\end{example}

\begin{example} 
Let $m=5$ and $\alpha$ be a generator of $\gf(2^m)^*$ with $\alpha^5 + \alpha^2 + 1=0$. Then 
the generator polynomial of the code $\C_s$ is 
\begin{eqnarray*} 
\m_s(x)=x^{16} + x^{15} + x^{13} + x^{12} + x^8 + x^6 + x^3 + 1
\end{eqnarray*}
and $\C_s$ is a $[31, 15, 8]$ binary cyclic code. Its dual is a $[31,16,7]$ cyclic code. Both codes are optimal. 
\end{example} 

\begin{example} 
Let $m=7$ and $\alpha$ be a generator of $\gf(2^m)^*$ with $\alpha^7 + \alpha + 1=0$. Then 
the generator polynomial of the code $\C_s$ is 
\begin{eqnarray*} 
\m_s(x) = & x^{36} + x^{34} + x^{33} + x^{32} + x^{29} + x^{28} + x^{27} +  
  x^{26} + x^{25} + \\ 
  &  x^{24} +  x^{21} + x^{12} + x^{11} + x^9 + x^7 + x^6 + x^5 + x^3 + x + 1  
\end{eqnarray*}
and $\C_s$ is a $[127, 91, 8]$ binary cyclic code. 
\end{example}

\subsection{Binary cyclic codes from the function $f(x)=x^{2^h-1}$}\label{sec-2hminus1} 

Functions over $\gf(2^m)$ of the form $f(x)=x^{2^h-1}$ may have good nonlinearity \cite{BCC}.  
Let $h$ be a positive integer satistying the following condition: 
\begin{eqnarray}\label{eqn-2m1cond} 
1 \le h \le \left\{ \begin{array}{l} 
                      (m-1)/2 \mbox{ if $m$ is odd and} \\
                      (m-2)/2 \mbox{ if $m$ is even.}   
                           \end{array} 
                           \right. 
\end{eqnarray}

In this subsection, we deal with the binary code $\C_s$ defined by the sequence $s^{\infty}$ 
of (\ref{eqn-sequence}), where $f(x)=x^{2^h-1}$. The binary code of this subsection can be 
viewed as a special case of the code of Section \ref{sec-qhminus1}. However, in the special case $q=2$ 
we are able to obtain better results, and will need the special techniques of this subsection to 
handle the codes of Section \ref{sec-Kasami}. Hence, we have to treat the code of this subsection 
here separately.   

We first prove a number of auxiliary results on $2$-cyclotomic cosets, which are stated in the 
following lemmas. 

\begin{lemma}\label{lem-2f261} 
For any $j$ with $1 \le j \le 2^h$, the size $\ell_j=|C_j|=m$. 
\end{lemma} 

\begin{proof} 
Let $1 \le j \le 2^h-1$, and let $u$ and $v$ be any two integers with $m-1\ge u > v \ge 0$. 
Define 
$$ 
\Delta(j, u, v) = j2^u-j2^v=2^vj(2^{u-v}-1). 
$$ 
Note that 
$$ 
\gcd(2^{u-v}-1, 2^m-1)=2^{\gcd(u-v,m)}-1. 
$$
If $u-v=1$, $\Delta(j,u,v)=2^vj$ is not divisble by $n=2^m-1$. We now consider the case 
that $m-1 \ge u-v\ge 2$. In this case, $\gcd(u-v,m) <m$. It then follows that 
$$ 
\gcd(2^{u-v}-1, 2^m-1) \le 2^{m/2}-1. 
$$
Note that 
$$ 
1 \le j \le 2^h-1 \le 2^{(m-1)/2}-1. 
$$
It then follows that 
$$ 
1 \le j \gcd(2^{u-v}-1, 2^m-1) \le (2^{m/2}-1)(2^{(m-1)/2}-1) < 2^m-1. 
$$
Hence $\Delta(j,u,v) \not\equiv 0 \pmod{n}$ in the case $m-1 \ge u-v \ge 2$. 
This completes the proof. 
\end{proof} 

\begin{lemma}\label{lem-2f262} 
For any pair of distinct odd $i$ and odd $j$ in the set $\{1,2, \cdots, 2^h-1\}$, 
$C_i \cap C_j = \emptyset$, i.e., they cannot be in the same $2$-cyclotomic 
coset modulo $n$.  
\end{lemma} 

\begin{proof} 
Define 
$$
\Delta_1=i 2^u -j \mbox{ and } \Delta_2=j 2^{m-u} -i.  
$$ 
Because $i$ and $j$ both are odd, $\Delta_i \ne 0$ for both $i$. 

Suppose that $i$ and $j$ are in the same cyclotomic coset. Then $n$ divides both 
$\Delta_1$ and $\Delta_2$. 

We distinguish between the following two cases. When $u \le h+1$, we have 
$$ 
-n<-(2^h-2) \le 1-j \le \Delta_1 \le (2^h-1)2^{h+1}-j < n. 
$$ 
In this case $\Delta_1 \not\equiv 0 \pmod{n}$. Hence, we have reached a 
contradiction. 

When $u \ge h+2$, we have $m-u \le m-h-2$ and 
$$ 
-n<-(2^h-2) \le 1-i \le \Delta_2 \le (2^h-1)2^{m-h-2}-j < n. 
$$ 
In this case $\Delta_2 \not\equiv 0 \pmod{n}$. Hence, we have also reached a 
contradiction. This completes the proof. 
\end{proof} 

We need to do more preparations before presenting and proving the main results 
of this subsection. Let $t$ be a positive integer. We define $T=2^t-1$. For any 
odd $a \in \{1,2,3,\cdots,T\}$ we define 
$$ 
\epsilon_a^{(t)} = \left\lceil \log_2 \frac{T}{a} \right\rceil 
$$
and 
$$ 
B_a^{(t)} =\left\{2^ia: i =0,1,2, \cdots, \epsilon_a^{(t)} -1 \right\}. 
$$
Then it can be verified that 
$$ 
\bigcup_{1 \le 2a+1 \le T} B_{2a+1}^{(t)} =\{1,2,3,\cdots, T\}
$$
and 
$$ 
B_a^{(t)} \cap B_b^{(t)} = \emptyset  
$$
for any pair of distinct odd numbers $a$ and $b$ in  $\{1,2,3,\cdots, T\}$. 

The following lemma follows directly from the definitions of $\epsilon_a^{(t)}$ 
and $B_a^{(t)}$. 

\begin{lemma}\label{lem-f263}
Let $a$ be an odd integer in $\{0,1,2. \cdots, T\}$. Then 
\begin{eqnarray*}
& & B_a^{(t+1)} = B_a^{(t)} \cup \{a 2^{\epsilon_a^{(t)}}\} \mbox{ if } 1 \le a \le 2^t-1, \\
& & B_a^{(t+1)} = \{a\} \mbox{ if } 2^t+1 \le a \le 2^{t+1}-1, \\
& & \epsilon_a^{(t+1)} = \epsilon_a^{(t)} +1  \mbox{ if } 1 \le a \le 2^t-1, \\
& & \epsilon_a^{(t+1)} = 1 \mbox{ if } 2^t+1 \le a \le 2^{t+1}-1. 
\end{eqnarray*} 
\end{lemma}

\begin{lemma}\label{lem-f264} 
Let $N_t$ denote the total number of odd  $\epsilon_a^{(t)}$ when $a$ ranges over all 
odd numbers in the set $\{1,2,\cdots, T\}$. Then $N_1=1$ and 
$$ 
N_t = \frac{2^t+(-1)^{t-1}}{3}
$$ 
for all $t \ge 2$. 
\end{lemma} 

\begin{proof} 
It is easily checked that $N_2=1$, $N_3=3$ and $N_4=5$. 
It follows from Lemma \ref{lem-f263} that 
$$ 
N_t= 2^{t-2} + (2^{t-2} -N_{t-1}). 
$$
Hence 
$$ 
N_t - 2^{t-2}  = 2^{t-3} - (N_{t-1}-2^{t-3})= 3\times 2^{t-4} + (N_{t-2}-2^{t-4}).
$$
With the recurcive application of this recurrence formula, one obtains the desired 
formula for $N_t$. 
\end{proof}

\begin{lemma}\label{lem-22mm1} 
Let $s^{\infty}$ be the sequence of (\ref{eqn-sequence}), where $f(x)=x^{2^h-1}$, $h\ge 2$ and $h$ satisfies the 
conditions of (\ref{eqn-2m1cond}). Then the linear span $\ls_s$ of $s^{\infty}$ is given by 
\begin{eqnarray}\label{eqn-22m0} 
\ls_s =\left\{ \begin{array}{l}
                   \frac{m(2^h+(-1)^{h-1})}{3} \mbox{ if $m$ is even} \\
                   \frac{m(2^h+(-1)^{h-1}) +3}{3} \mbox{ if $m$ is odd.}                    
\end{array}
\right. 
\end{eqnarray} 
We have then 
\begin{equation}\label{eqn-2m31}
\m_s(x) =
 (x-1)^{\N_2(m)} \prod_{1 \le 2j+1 \le 2^h-1 \atop \epsilon_{2j+1}^{(h)} \bmod{2}=1} m_{\alpha^{-(2j+1)}}(x),                   
\end{equation} 
where $m_{\alpha^{-j}}(x)$ is the minimal polynomial of $\alpha^{-j}$ over $\gf(2)$. 
\end{lemma} 

\begin{proof} 
We have 
\begin{eqnarray}\label{eqn-22m11}
\tr(f(x+1))  
&=& \tr\left(  (x+1)^{\sum_{i=0}^{h-1} 2^{i}}    \right) \nonumber \\
&=& \tr\left( \prod_{i=0}^{h-1} \left(x^{2^{i}}+1\right)     \right) \nonumber \\
&=& \tr\left( \sum_{i=0}^{2^h-1} x^{i}     \right) \nonumber \\
&=&\tr(1) + \tr\left( \sum_{i=1}^{2^h-1} x^{i}     \right) \nonumber \\
&=& \tr(1) + \tr\left( \sum_{1 \le 2i+1 \le 2^h-1 \atop \epsilon_{2i+1}^{(h)} \bmod{2}=1} x^{2i+1}     \right)  
\end{eqnarray} 
where the last equality follows from Lemma \ref{lem-2f262}. 

By definition, the sequence of (\ref{eqn-sequence}) is given by $s_t=\tr(f(\alpha^t+1))$ for all $t \ge 0$. 
The desired conclusions on the linear span and the minimal polynomial $\m_s(x)$ then follow from Lemmas \ref{lem-ls2}, 
\ref{lem-2f261}, \ref{lem-2f262}, \ref{lem-f264} and Equation   
(\ref{eqn-22m11}). 
\end{proof}

The following theorem provides information on the code $\C_{s}$.    

\begin{theorem}\label{thm-38} 
Let $h \ge 2$. 
The binary code $\C_{s}$ defined by the binary sequence of Lemma \ref{lem-22mm1} has parameters 
$[2^m-1, 2^{m}-1-\ls_s, d]$ and generator polynomial $\m_s(x)$ of  (\ref{eqn-2m31}), 
 where $\ls_s$ is  given in (\ref{eqn-22m0}) and 
 \begin{eqnarray*} 
 d \ge \left\{ \begin{array}{l} 
                     2^{h-2}+2 \mbox{ if $m$ is odd and $h>2$} \\
                     2^{h-2}+1.  
                     \end{array}  
 \right.  
 \end{eqnarray*}   
\end{theorem} 

\begin{proof} 
The dimension of $\C_{s}$ follows from Lemma \ref{lem-22mm1} and the definition of the 
code $\C_s$. We now derive the lower bounds on the minimum weight $d$ of the code. It is 
well known that the codes generated by $\m_s(x)$ and its reciprocal have the same weight 
distribution. It follows from Lemmas \ref{lem-22mm1}  and \ref{lem-f263} that the reciprocal 
of $\m_s(x)$ has zeros $\alpha^{2j+1}$ for all $j$ in $\{2^{h-2}, 2^{h-2}+1, \cdots, 2^{h-1}-1\}$. 
By the Hartman-Tzeng bound, we have $d \ge 2^{h-2}+1$. If $m$ is odd, $\C_s$ is an 
even-weight code. In this case, $d \ge 2^{h-2}+2$. 
\end{proof}

\begin{example} 
Let $(m,h)=(3,2)$ and $\alpha$ be a generator of $\gf(2^m)^*$ with $\alpha^3 + \alpha + 1=0$. Then 
$\C_s$ is a $[7, 3, 4]$ binary code with generator polynomial  
$$ 
\m_s(x)=x^4 + x^3 + x^2 + 1.    
$$
Its dual is a $[7,4,3]$ cyclic code. Both codes are optimal. 
\end{example}

\begin{example} 
Let $(m,h)=(5,2)$ and $\alpha$ be a generator of $\gf(2^m)^*$ with $\alpha^5 + \alpha^2 + 1=0$. Then 
the generator polynomial of the code $\C_s$ is 
\begin{eqnarray*} 
\m_s(x)=x^6 + x^5 + x^4 + 1
\end{eqnarray*}
and $\C_s$ is a $[31, 25, 4]$ binary cyclic code and optimal. 
\end{example} 

\begin{example} 
Let $(m,h)=(7,2)$ and $\alpha$ be a generator of $\gf(2^m)^*$ with $\alpha^7 + \alpha + 1=0$. Then 
the generator polynomial of the code $\C_s$ is 
\begin{eqnarray*} 
\m_s(x) = x^8 + x^6 + x^5 + x^4 + x^3 + x^2 + x + 1 
\end{eqnarray*}
and $\C_s$ is a $[127, 119, 4]$ binary cyclic code and optimal. 
\end{example} 

\begin{example} 
Let $(m, h)=(7,3)$ and $\alpha$ be a generator of $\gf(2^m)^*$ with $\alpha^7 + \alpha + 1=0$. Then 
the generator polynomial of the code $\C_s$ is 
\begin{eqnarray*} 
\m_s(x) &=& x^{22} + x^{21} + x^{20} + x^{18} + x^{17} + x^{16} + x^{14} + \\ 
& & x^{13} + x^8 + x^7 + x^6 + x^5 + x^4 + 1
\end{eqnarray*}
and $\C_s$ is a $[127, 105, d]$ binary cyclic code, where $4 \le d \le 8$.  
\end{example} 

\begin{remark} 
The code $\C_s$ of Theorem \ref{thm-38} may be bad when $\gcd(h, m) \ne 1$. In this case the function 
$f(x)=x^{2^h-1}$ is not a permutation of $\gf(2^m)$. For example, when $(m, h)=(6,3)$, $\C_s$ is a 
$[63, 45, 3]$ binary cyclic code, while the best known linear code has parameters $[63, 45, 8]$.  
\end{remark}

\subsection{Binary cyclic codes from the first Niho APN function }\label{sec-1Niho} 

The first Niho APN function is defined by $f(x)=x^e$, where $e=2^{(m-1)/2}+2^{(m-1)/4}-1$ and $m \equiv 1 \pmod{4}$. 
Define $h=(m-1)/4$. We have then  
\begin{eqnarray}\label{eqn-1Niho}
\tr(f(x+1)) 
&=& \tr\left( (x^{2^{2h}}+1) (x+1) ^{\sum_{i=0}^{h-1} 2^{i}}    \right) \nonumber \\
&=& \tr\left( (x^{2^{2h}}+1) \prod_{i=0}^{h-1} \left(x^{2^{i}}+1\right)     \right) \nonumber \\
&=& \tr\left( (x^{2^{2h}}+1) \sum_{i=0}^{2^h-1} x^{i}     \right) \nonumber \\ 
&=& 1+ \tr\left(\sum_{i=0}^{2^h-1} x^{i+2^{2h}}  + \sum_{i=1}^{2^h-1} x^{i}     \right). 
\end{eqnarray} 

The sequence $s^{\infty}$ of (\ref{eqn-sequence}) defined by the first Niho function is then 
given by 
\begin{eqnarray}\label{eqn-1Nihoseq}
s_t= 1+ \tr\left(\sum_{i=0}^{2^h-1} (\alpha^t)^{i+2^{2h}}  + \sum_{i=1}^{2^h-1} (\alpha^t)^{i}     \right)
\end{eqnarray} 
for all $t \ge 0$, where $\alpha$ is a generator of $\gf(2^m)^*$.  
In this subsection, we deal with the code $\C_s$ defined by the sequence $s^{\infty}$ of 
(\ref{eqn-1Nihoseq}). To this end, we need to prove a number of auxilary results on 
$2$-cyclotomic cosets. 

We define the following two sets for convenience:
\begin{eqnarray*} 
A=\{0,1,2, \cdots, 2^h-1\}, \ B=2^{2h}+A=\{i+2^{2h}: i \in A\}.  
\end{eqnarray*} 

\begin{lemma}\label{lem-1Nf261} 
For any $j \in B$, the size $\ell_j=|C_j|=m$. 
\end{lemma} 

\begin{proof} 
Let $j = i + 2^{2h}$, where $i \in A$. For any $u$ with $1 \le u \le m-1$, define  
\begin{eqnarray*} 
& & \Delta_1(j, u) = j(2^{u}-1)=(i+2^{2h}) (2^{u}-1), \\
& & \Delta_2(j, u) = j(2^{m-u}-1)=(i+2^{2h}) (2^{m-u}-1). 
\end{eqnarray*} 
If $\ell_j <m$, there would be an integer $1 \le u \le m-1$ such that $\Delta_t(j,u) \equiv 0 \pmod{n}$ 
for all $t \in \{1,2\}$.  

Note that $1 \le u \le m-1$. We have that $\Delta_1(j, u) \ne 0$ and  $\Delta_2(j, u) \ne 0$. 
When $u \le m-2h-1$, we have  
$$ 
2^h \le \Delta_1(j, u) \le (2^{2h}+2^h-1)(2^{m-2h-1}-1) <n. 
$$
In this case, $\Delta_1(j,u) \not\equiv 0 \pmod{n}$.  

When $u \ge m-2h$, we have $m-u \le 2h$ and  
$$ 
2^h \le \Delta_2(j, u) \le (2^{2h}+2^h-1)(2^{2h}-1) <n. 
$$
In this case, $\Delta_2(j,u) \not\equiv 0 \pmod{n}$.  

Combining the conclusions of the two cases above completes the proof. 
\end{proof} 

\begin{lemma}\label{lem-1Nf262} 
For any pair of distinct $i$ and $j$ in $B$, 
$C_i \cap C_j = \emptyset$, i.e., they cannot be in the same $2$-cyclotomic 
coset modulo $n$.  
\end{lemma} 

\begin{proof} 
Let $i=i_1+2^{2h}$ and $j=j_1+2^{2h}$, where $i_1 \in A$ and $j_1 \in A$. Define 
\begin{eqnarray*} 
& & \Delta_1(i, j, u) = i2^{u}-j=  (i_1+2^{2h})2^u-  (j_1+2^{2h}), \\
& & \Delta_2(i, j, u) = j2^{m-u}-i= (j_1+2^{2h})2^{m-u}-  (i_1+2^{2h}). 
\end{eqnarray*} 
If $C_i=C_j$, there would be an integer $1 \le u \le m-1$ such that $\Delta_t(i,j,u) \equiv 0 \pmod{n}$ 
for all $t \in \{1,2\}$.

We first prove that $\Delta_1(i, j, u) \ne 0$. When $u=0$,  $\Delta_1(i, j, u)=i_1-j_1 \ne 0$. When 
$1 \le u \le m-1$, we have 
$$ 
\Delta_1(i, j, u) \ge 2i_1 + 2^{2h+1}-2^{2h}-j_1 >0. 
$$  

Since $1 \le u \le m-1$, one can similarly prove that $\Delta_2(i, j, u) >0$.

When $u \le m-2h-1$, we have  
$$ 
-n < -2^{2h} \le \Delta_1(i,j, u) \le (2^{2h}+2^h-1)(2^{m-2h-1}-1) <n. 
$$
In this case, $\Delta_1(i,j,u) \not\equiv 0 \pmod{n}$.  

When $u \ge m-2h$, we have $m-u \le 2h$ and  
$$ 
0< \Delta_2(i, j, u) \le (2^{2h}+2^h-1)2^{2h}-i_1-2^h <n. 
$$
In this case, $\Delta_2(i,j,u) \not\equiv 0 \pmod{n}$.  

Combining the conclusions of the two cases above completes the proof. 
\end{proof} 

\begin{lemma}\label{lem-feb281}
For any $i+2^{2h} \in B$ and odd $j \in A$, 
\begin{eqnarray}
C_{i+2^{2h}} \cap C_j = \left\{ \begin{array}{l}
                             C_j \mbox{ if } (i,j)=(0,1) \\
                             \emptyset \mbox{ otherwise.}  
\end{array}
\right. 
\end{eqnarray}
\end{lemma}  

\begin{proof} 
Define 
\begin{eqnarray*} 
& & \Delta_1(i, j, u) = j2^{u}-(i+2^{2h}),  \\
& & \Delta_2(i, j, u) =  (i+2^{2h})2^{m-u}-  j. 
\end{eqnarray*} 
Suppose $C_{i+2^{2h}}=C_j$, there would be an integer $0 \le u \le m-1$ such that $\Delta_t(i,j,u) \equiv 0 \pmod{n}$ 
for all $t \in \{1,2\}$.  

If $u=2h$, then 
\begin{eqnarray*}
0 \equiv \Delta_1(i, j, u) & \equiv & 2^{2h+1}(j2^{2h}-(i+2^{2h})) \pmod{n} \\ 
 & \equiv & j2^{m}-i2^{2h+1}-2^{m} \pmod{n} \\ 
 & \equiv &  j -1 -i 2^{2h+1} \pmod{n} \\
 & = &  j -1 -i 2^{2h+1}.    
\end{eqnarray*}
Whence, the only solution of $\Delta_1(i,j,2h) \equiv 0 \pmod{n}$ is $(i,j)=(0,1)$.   

We now consider the case that $0 \le u <2h$. We claim that $\Delta_1(i, j, u) \ne 0$. 
Suppose on the contrary that  $\Delta_1(i, j, u) = 0$. We would then have  
$$ 
j 2^u -i - 2^{2h} =0. 
$$ 
Because $u<2h$ and $j$ is odd, there is an odd $i_1$ such that $i=2^u i_1$. It then 
follows from $i < 2^h$ that $u <h$. We obtain then 
$$ 
j=i_1+2^{2h-u}>i_1 + 2^{h}>2^h-1. 
$$ 
This is contrary to the assumption that $j \in A$. This proves that $\Delta_1(i, j, u) \ne 0$. 

Finally, we deal with the case that $2h+1 \le u <4h=m-1$. We prove that $\Delta_2(i, j, u) \not\equiv 0 \pmod{n}$ 
in this case.  Since $j$ is odd, $\Delta_2(i, j, u) \ne 0$. We have also 
\begin{eqnarray*}
\Delta_2(i, j, u) 
&=& i2^{m-u}+2^{m+2h-u} -j \\
&\le & (2^h-1)2^{m-u} + 2^{m-1} -j \\
&\le & 2^{m-(h-1)}+2^{m-1}-j \\
&<& n. 
\end{eqnarray*} 
Clearly, $\Delta_2(i, j, u) >-j >-n$. Hence in this case we have $\Delta_2(i, j, u) \not\equiv 0 \pmod{n}$. 

Summarizing the conclusions above proves this lemma. 
\end{proof}

\begin{lemma}\label{lem-1N2m1} 
Let $m \ge 9$ be odd. 
Let $s^{\infty}$ be the sequence of (\ref{eqn-1Nihoseq}). Then the linear span $\ls_s$ of $s^{\infty}$ is given by 
\begin{eqnarray}\label{eqn-1N2m0} 
\ls_s =\left\{ \begin{array}{l}
                   \frac{m\left(2^{(m+7)/4}+(-1)^{(m-5)/4}\right) +3}{3} \mbox{ if $m \equiv 1 \pmod{8}$} \\
                   \frac{m\left(2^{(m+7)/4}+(-1)^{(m-5)/4}-6\right) +3}{3} \mbox{ if $m \equiv 5 \pmod{8}$.}                    
\end{array}
\right. 
\end{eqnarray} 
We have also 
\begin{equation*}\label{eqn-1N2m21}
\m_s(x) = (x-1) \prod_{i=0}^{2^{\frac{m-1}{4}}-1} m_{\alpha^{-i-2^{\frac{m-1}{2}} }}(x)   
\prod_{1 \le 2j+1 \le 2^{\frac{m-1}{4}}-1 \atop \epsilon_{2j+1}^{((m-1)/4)} \bmod{2}=1} m_{\alpha^{-2j-1}}(x)               
\end{equation*} 
if $m \equiv 1 \pmod{8}$; and 
\begin{equation*}\label{eqn-1N2m31}
\m_s(x) =(x-1) \prod_{i=1}^{2^{\frac{m-1}{4}}-1} m_{\alpha^{-i-2^{\frac{m-1}{2}} }}(x)   
\prod_{3 \le 2j+1 \le 2^{\frac{m-1}{4}}-1 \atop \epsilon_{2j+1}^{((m-1)/4)} \bmod{2}=1} m_{\alpha^{-2j-1}}(x)  
\end{equation*}
if $m \equiv 5 \pmod{8}$, 
where $m_{\alpha^{-j}}(x)$ is the minimal polynomial of $\alpha^{-j}$ over $\gf(2)$ and $\epsilon_{2j+1}^{(h)}$ was 
defined in Section \ref{sec-Inverse}. 
\end{lemma} 

\begin{proof} 
By Lemma \ref{lem-1Nf262}, the monomials in the  function 
\begin{equation}\label{eqn-feb28111}
\tr\left(\sum_{i=0}^{2^h-1} x^{i+2^{2h}} \right)  
\end{equation}
will not cancel each other. Lemmas \ref{lem-f264} and \ref{lem-22mm1} say that after cancellation, we have 
\begin{equation}\label{eqn-feb28121}
\tr\left(\sum_{i=1}^{2^h-1} x^{i} \right) = 
\tr\left(\sum_{1 \le 2j+1 \le 2^h-1 \atop \epsilon_{2j+1}^{(h)}} x^{2j+1} \right). 
\end{equation} 

By Lemma \ref{lem-feb281}, the monomials in the function of (\ref{eqn-feb28111}) will not cancel the monomials in 
the function in the right-hand side of (\ref{eqn-feb28121}) if $m \equiv 1 \pmod{8}$, and only the term 
$x^{2^{2h}}$ in the function of (\ref{eqn-feb28111}) cancels the monomial $x$ in the function in the right-hand 
side of (\ref{eqn-feb28121}) if $m \equiv 5 \pmod{8}$.  

The desired conclusions on the linear span and the minimal polynomial $\m_s(x)$ then follow from Lemmas \ref{lem-ls2}, 
\ref{lem-1Nf261},  and Equation   
(\ref{eqn-1Niho}). 
\end{proof}

The following theorem provides information on the code $\C_{s}$.    

\begin{theorem}\label{thm-yue} 
Let $m \ge 9$ be odd. 
The binary code $\C_{s}$ defined by the sequence of (\ref{eqn-1Nihoseq}) has parameters 
$[2^m-1, 2^{m}-1-\ls_s, d]$ and generator polynomial $\m_s(x)$,  
where $\ls_s$ and $\m_s(x)$ are given in Lemma \ref{lem-1N2m1} and the minimum weight $d$ has the following 
bounds: 
\begin{eqnarray}\label{eqn-niho1b}
d \ge \left\{ \begin{array}{ll} 
 2^{(m-1)/4} + 2 & \mbox{if } m \equiv 1 \pmod{8} \\
 2^{(m-1)/4}       & \mbox{if } m \equiv 5 \pmod{8}. 
\end{array} 
\right. 
\end{eqnarray}  
\end{theorem} 

\begin{proof} 
The dimension  and the generator polynomial of $\C_{s}$ follow from Lemma \ref{lem-1N2m1} and the 
definition of the code $\C_s$.  We now derive the lower bounds on the minimum weight $d$. It is well 
known that the codes generated by $\m_s(x)$ and its reciprocal have the same weight distribution. The 
reciprocal of $\m_s(x)$ has the zeros $\alpha^{i+2^{2h}}$ for all $i$ in $\{0,1,2, \cdots, 2^h-1\}$ if 
$m \equiv 1 \pmod{8}$, and for all $i$ in $\{1,2, \cdots, 2^h-1\}$ if $m \equiv 5 \pmod{8}$. 
Note that $\C_s$ is an even-weight code.  Then the desired bounds on $d$ follow from the BCH bound.  
\end{proof} 

\begin{example} 
Let $m=5$ and $\alpha$ be a generator of $\gf(2^m)^*$ with $\alpha^5 + \alpha^2 + 1=0$. Then 
the generator polynomial of the code $\C_s$ is 
\begin{eqnarray*} 
\m_s(x)=x^6 + x^3 + x^2 + 1
\end{eqnarray*}
and $\C_s$ is a $[31, 25, 4]$ binary cyclic code and  optimal. 

\end{example} 

\begin{example} 
Let $m=9$ and $\alpha$ be a generator of $\gf(2^m)^*$ with $\alpha^9 + \alpha^4 + 1=0$. Then 
the generator polynomial of the code $\C_s$ is 
\begin{eqnarray*} 
\m_s(x) &= & x^{46} + x^{45} + x^{41} + x^{40} + x^{39} + x^{36} + x^{35} + \\ 
& & x^{33} + x^{28} + x^{27} + x^{26} + x^{25} + x^{24} + x^{22} + x^{21} + \\
& & x^{20} + x^{19} + x^{14} + x^{12} + x^7 + x^4 +
    x^2 + x + 1 
\end{eqnarray*}
and $\C_s$ is a $[511, 465, d]$ binary cyclic code, where $d \ge 6$. The actual minimum weight 
may be larger than 6. 

\end{example}

\subsection{Binary cyclic codes from the Kasami APN function }\label{sec-Kasami} 

The Kasami APN function is defined by $f(x)=x^e$, where $e=2^{2h}-2^h+1$ and $\gcd(m,h)=1$. 
The case $h=1$ is covered by the Gold APN function. 
In this subsection, we have the following additional restrictions on $h$: 
\begin{eqnarray}\label{eqn-Hcondition}
2 \le h \le \left\{  \begin{array}{l}
                   \frac{m-1}{4} \mbox{ if } m \equiv 1 \pmod{4}, \\
                   \frac{m-3}{4} \mbox{ if } m \equiv 3 \pmod{4}, \\                   
                   \frac{m-4}{4} \mbox{ if } m \equiv 0 \pmod{4}, \\
                   \frac{m-2}{4} \mbox{ if } m \equiv 2 \pmod{4}.                    
\end{array} 
\right. 
\end{eqnarray}

Note that 
\begin{eqnarray}\label{eqn-Kasami}
\tr(f(x+1)) 
&=& \tr\left( (x+1) (x+1) ^{\sum_{i=0}^{h-1} 2^{h+i}}    \right) \nonumber \\
&=& \tr\left( (x+1) \prod_{i=0}^{h-1} \left(x^{2^{h+i}}+1\right)     \right) \nonumber \\
&=& \tr\left( (x+1) \sum_{i=0}^{2^h-1} x^{2^{h}i}     \right) \nonumber \\
&=& \tr\left(\sum_{i=0}^{2^h-1} x^{2^{h}i+1}  + \sum_{i=0}^{2^h-1} x^{2^{h}i}     \right)  \nonumber \\ 
&=& \tr\left(\sum_{i=0}^{2^h-1} x^{2^{h}i+1}  + \sum_{i=0}^{2^h-1} x^{i}     \right) \nonumber \\
&=& \tr\left(\sum_{i=0}^{2^h-1} x^{i+2^{m-h}}  + \sum_{i=1}^{2^h-1} x^{i}     \right) +1. 
\end{eqnarray} 

The sequence $s^{\infty}$ of (\ref{eqn-sequence}) defined by the Kasami function is then 
\begin{eqnarray}\label{eqn-Kasamiseq}
s_t= \tr\left(\sum_{i=0}^{2^h-1} (\alpha^t)^{i+2^{m-h}}  + \sum_{i=1}^{2^h-1} (\alpha^t)^{i}     \right) +1
\end{eqnarray} 
for all $t \ge 0$, where $\alpha$ is a generator of $\gf(2^m)^*$. 

In this subsection, we deal with the code $\C_s$ defined by the sequence $s^{\infty}$ of 
(\ref{eqn-Kasamiseq}). 
It is noticed that the final expression of the function of (\ref{eqn-Kasami}) is of the same format 
as that of the function of (\ref{eqn-1Niho}). The proofs of the lemmas and theorems in this subsection 
are very similar to those of Section \ref{sec-2hminus1}. Hence, we present only the main results without 
providing proofs.    

We define the following two sets for convenience: 
\begin{eqnarray*} 
A=\{0,1,2, \cdots, 2^h-1\}, \ B=2^{m-h}+A=\{i+2^{m-h}: i \in A\}.  
\end{eqnarray*} 

\begin{lemma}\label{lem-K1Nf261} 
Let $h$ satisfy the conditions of (\ref{eqn-Hcondition}). 
For any $j \in B$, the size $\ell_j=|C_j|=m$. 
\end{lemma} 

\begin{proof} 
The proof of Lemma \ref{lem-1Nf261} is easily modified into a proof for this lemma. 
The detail is left to the reader. 
\end{proof} 

\begin{lemma}\label{lem-K1Nf262} 
Let $h$ satisfy the conditions of (\ref{eqn-Hcondition}). 
For any pair of distinct $i$ and $j$ in $B$, 
$C_i \cap C_j = \emptyset$, i.e., they cannot be in the same $2$-cyclotomic 
coset modulo $n$.  
\end{lemma} 

\begin{proof} 
The proof of Lemma \ref{lem-1Nf262} is easily modified into a proof for this lemma. 
The detail is left to the reader. 
\end{proof} 

\begin{lemma}\label{lem-Kfeb281} 
Let $h$ satisfy the conditions of (\ref{eqn-Hcondition}). 
For any $i+2^{m-h} \in B$ and odd $j \in A$, 
\begin{eqnarray}
C_{i+2^{m-h}} \cap C_j = \left\{ \begin{array}{l}
                             C_j \mbox{ if } (i,j)=(0,1) \\
                             \emptyset \mbox{ otherwise.}  
\end{array}
\right. 
\end{eqnarray}
\end{lemma}  

\begin{proof} 
The proof of Lemma \ref{lem-feb281} is easily modified into a proof for this lemma. 
The detail is left to the reader. 
\end{proof}

\begin{lemma}\label{lem-K1N2m1} 
Let $h$ satisfy the conditions of (\ref{eqn-Hcondition}). 
Let $s^{\infty}$ be the sequence of (\ref{eqn-Kasamiseq}). Then the linear span $\ls_s$ of $s^{\infty}$ is given by 
\begin{eqnarray}\label{eqn-K1N2m0} 
\ls_s =\left\{ \begin{array}{l}
                   \frac{m\left(2^{(h+2}+(-1)^{h-1}\right) +3}{3} \mbox{ if $h$ is even} \\
                   \frac{m\left(2^{h+2}+(-1)^{h-1}-6\right) +3}{3} \mbox{ if $h$ is odd.}                    
\end{array}
\right. 
\end{eqnarray} 
We have also 
\begin{equation*}\label{eqn-K1N2m21}
\m_s(x) = (x-1) \prod_{i=0}^{2^{h}-1} m_{\alpha^{-i-2^{m-h} }}(x)   
\prod_{1 \le 2j+1 \le 2^{h}-1 \atop \epsilon_{2j+1}^{h} \bmod{2}=1} m_{\alpha^{-2j-1}}(x)               
\end{equation*} 
if $h$ is even; and 
\begin{equation*}\label{eqn-K1N2m31}
\m_s(x) =(x-1) \prod_{i=1}^{2^{h}-1} m_{\alpha^{-i-2^{m-h} }}(x)   
\prod_{3 \le 2j+1 \le 2^{h}-1 \atop \epsilon_{2j+1}^{h} \bmod{2}=1} m_{\alpha^{-2j-1}}(x)  
\end{equation*}
if $h$ is odd, 
where $m_{\alpha^{-j}}(x)$ is the minimal polynomial of $\alpha^{-j}$ over $\gf(2)$ and $\epsilon_{2j+1}^{(h)}$ was 
defined in Section \ref{sec-Inverse}. 
\end{lemma} 

\begin{proof} 
The proof of Lemma \ref{lem-1N2m1} is easily modified into a proof for this lemma. 
The detail is left to the reader. 
\end{proof}

The following theorem provides information on the code $\C_{s}$.    

\begin{theorem} \label{thm-Kyue}
Let $h$ satisfy the conditions of (\ref{eqn-Hcondition}). 
The binary code $\C_{s}$ defined by the sequence of (\ref{eqn-Kasamiseq}) has parameters 
$[2^m-1, 2^{m}-1-\ls_s, d]$ and generator polynomial $\m_s(x)$, 
where $\ls_s$ and $\m_s(x)$ are given in Lemma \ref{lem-K1N2m1} and the minimum weight $d$ has the following 
bounds: 
\begin{eqnarray}\label{eqn-Kniho1b}
d \ge \left\{ \begin{array}{ll} 
 2^{h} + 2 & \mbox{if $h$ is even}  \\
 2^{h}       & \mbox{if $h$ is odd.} 
\end{array} 
\right. 
\end{eqnarray}  
\end{theorem} 

\begin{proof} 
The proof of Lemma \ref{thm-yue} is easily modified into a proof for this lemma with the helps of the 
lemmas presented in this subsection. 
The detail is left to the reader. 
\end{proof}

\begin{example} 
Let $(m,h)=(3,2)$ and $\alpha$ be a generator of $\gf(2^m)^*$ with $\alpha^3 + \alpha + 1=0$. Then 
$\C_s$ is a $[7, 3, 4]$ binary code with generator polynomial  
$$ 
\m_s(x)=x^4 + x^3 + x^2 + 1.    
$$
Its dual is a $[7,4,3]$ cyclic code. Both codes are optimal. In this example, the condition of (\ref{eqn-Hcondition}) 
is not satisfied. So the conclusions on the code of 
this example may not agree with the conclusions of Theorem \ref{thm-Kyue}. 
\end{example}

\begin{example} 
Let $(m,h)=(5,2)$ and $\alpha$ be a generator of $\gf(2^m)^*$ with $\alpha^5 + \alpha^2 + 1=0$. Then 
the generator polynomial of the code $\C_s$ is 
\begin{eqnarray*} 
\lefteqn{\m_s(x)=} \\
& x^{16} + x^{14} + x^{10} + x^{9} + x^8 + x^7 + x^5 + x^4 + x^3 + x^2 + x+ 1
\end{eqnarray*}
and $\C_s$ is a $[31, 15, 8]$ binary cyclic code. Its dual is a $[31,16,7]$ cyclic code. Both codes are optimal. 
In this example, the condition of (\ref{eqn-Hcondition}) is not satisfied. So the conclusions on the code of 
this example do not agree with the conclusions of Theorem \ref{thm-Kyue}. 
\end{example} 

\begin{example} 
Let $(m,h)=(7,2)$ and $\alpha$ be a generator of $\gf(2^m)^*$ with $\alpha^7 + \alpha + 1=0$. Then 
the generator polynomial of the code $\C_s$ is 
\begin{eqnarray*} 
\m_s(x) &= & x^{36} + x^{28} + x^{27} + x^{23} + x^{21} + x^{20} + x^{18} + \\ 
& &  x^{13} + x^{12} + x^9 + x^7 + x^6 + x^5  + 1  
\end{eqnarray*}
and $\C_s$ is a $[127, 91, 8]$ binary cyclic code. 

\end{example} 

In this subsection, we obtained interesting results on the code $\C_s$ under the conditions 
of (\ref{eqn-Hcondition}). When $h$ is outside the ranges, it may be hard to determine the 
dimension of the code $\C_s$, let alone the minimum weight $d$ of the code. Hence, 
it would be nice if the following open problem can be solved. 

\begin{open} 
Determine the dimension and the minimum weight of the code $\C_s$ defined by the Kasami 
APN power function when $h$ satisfies 
\begin{eqnarray}\label{eqn-Hcondition2} 
\left\{  \begin{array}{l}
                 \frac{m-1}{2} \ge h >                  \frac{m-1}{4} \mbox{ if } m \equiv 1 \pmod{4}, \\
                 \frac{m-3}{2} \ge h >                   \frac{m-3}{4} \mbox{ if } m \equiv 3 \pmod{4}, \\                   
                 \frac{m-4}{2} \ge h >                   \frac{m-4}{4} \mbox{ if } m \equiv 0 \pmod{4}, \\
                 \frac{m-2}{2} \ge h >                   \frac{m-2}{4} \mbox{ if } m \equiv 2 \pmod{4}.                    
\end{array} 
\right. 
\end{eqnarray}
\end{open}

\subsection{Open problems regarding binary cyclic codes from APN functions}

In the previous subsections of Section \ref{sec-binary}, we investigated binary cyclic codes from 
some APN functions. It would be good if the following open problems could be solved. 

\begin{open} 
Determine the dimension and the minimum weight of the code $\C_s$ defined by the second  
Niho APN function $x^e$, where $e=2^{(m-1)/2}+2^{(3m-1)/4}-1$ and $m \equiv 3 \pmod{4}$.  
\end{open}  

\begin{open} 
Determine the dimension and the minimum weight of the code $\C_s$ defined by the Dobbertin   
APN function $x^e$, where $e=2^{4i}+2^{3i}+2^{2i}+2^i-1$ and $m=5i$.  
\end{open}  

A number of other types of APN functions $f(x)$ were discovered in \cite{BC,BCL2,BCL1,BCP}. 
The code $\C_s$ defined by these APN functions may also have good parameters. It would be 
interesting to investigate these codes.      

\section{Nonbinary cyclic codes from APN and planar functions}\label{sec-nonbinary} 

In this section, $q$ is a power of an odd prime $p$. Note that both APN and planar 
functions on $\gf(q^m)$ exist. 

\subsection{Cyclic codes from the planar function $x^2$}\label{sec-2square} 

In this subsection we study the code $\C_s$ defined by the planar function $f(x)=x^{2}$ 
on $\gf(r)$. To this end, we need to prove the following lemma.  

\begin{lemma}\label{lem-2Gold} 
Let $s^{\infty}$ be the sequence of (\ref{eqn-sequence}), where $f(x)=x^2$. 
Then the linear span $\ls_s$ of $s^{\infty}$ is equal to $2m+\N_p(m)$ and the minimal polynomial $\m_s(x)$ 
of  $s^{\infty}$ is given by 
\begin{equation}\label{eqn-2Gold}
\m_s(x)= (x-1)^{\N_p(m)} m_{\alpha^{-1}}(x) m_{\alpha^{-2}}(x)
\end{equation} 
where $m_{\alpha^{-j}}(x)$ is the minimal polynomial of $\alpha^{-j}$ over $\gf(q)$, 
$\N_p(i) =0$ if $i \equiv 0 \pmod{p}$ and $\N_p(i) =1$ otherwise. 
\end{lemma} 

\begin{proof} 
It is easily seen that 
\begin{eqnarray}\label{eqn-2Gold2}
s_t &=& \tr((\alpha^{t}+1)^{2}) \nonumber \\ 
&=& \tr(1)+  2\sum_{j=0}^{m-1} (\alpha^t)^{q^j} + \sum_{j=0}^{m-1} (\alpha^t)^{2q^j}. 
\end{eqnarray}

It can be easily proved that $\ell_1=\ell_{n-1}=\ell_{2}=\ell_{n-2}=m$ and $C_1 \cap C_2=\emptyset$. 
The desired conclusions on the linear span and the minimal polynomial $\m_s(x)$ then follow from Lemma \ref{lem-ls2} 
and (\ref{eqn-2Gold2}). 
\end{proof}

The following theorem provides information on the code $\C_{s}$.    

\begin{theorem}\label{thm-2square} 
The code $\C_{s}$ defined by the sequence of Lemma \ref{lem-2Gold} has parameters 
$[n, n-2m-\N_p(m), d]$ and generator polynomial $\m_s(x)$ of (\ref{eqn-2Gold}),  
where 
\begin{eqnarray*} 
\left\{ \begin{array}{ll}
d=4             & \mbox{ if $q=3$ and $\N_p(m)=0$,} \\
4 \le d \le 5 & \mbox{ if $q=3$ and $\N_p(m)=1$,} \\  
3 \le d \le 4 & \mbox{ if $q >3$ and $\N_p(m)=1$,} \\ 
d = 3           & \mbox{ if $q >3$ and $\N_p(m)=0$.} 
\end{array} 
\right. 
\end{eqnarray*}  
\end{theorem} 

\begin{proof} 
The dimension of $\C_{s}$ follows from Lemma \ref{lem-2Gold} and the definition of the 
code $\C_s$.  We need to prove the conclusion on the minimum weight $d$ of $\C_{s}$. 
The code $\C_s$ of this theorem has the same weight distribution as the dual code in Theorem 
11 in \cite{CDY05} when $m \equiv 0 \pmod{p}$ and has the same weight distribution as the 
even-like subcode of  the dual code in Theorem 11 in \cite{CDY05} when $m \not\equiv 0 \pmod{p}$.  
This is because the generator polynomial of the code $\C_s$ of this theorem and that of the  
dual code in Theorem 11 in \cite{CDY05} are reciprocals of each other when $m \equiv 0 \pmod{p}$, 
and the generator polynomial of the code $\C_s$ of this theorem and that of the even-like subcode 
of the 
dual code in Theorem 11 in \cite{CDY05} are reciprocals of each other when $m \not\equiv 0 \pmod{p}$.  
The conclusions on $d$ then follow from those on the minimum weight of the dual code in 
Theorem 11 in \cite{CDY05}.  
\end{proof} 

\begin{remark} 
Both this paper and \cite{CDY05} employ planar functions in constructing codes. The two 
approaches are different in general. The approach of this paper always produces cyclic 
codes. The approach of \cite{CDY05} gives cyclic codes only when the planar function 
$f(x)$ is a power function. Another difference is that the two approaches produce codes 
with different dimensions. However, for this quadratic planar function $f(x)=x^2$, the 
codes obtained with the two approaches are closely related. The relation between the two 
codes is made clear in the proof of  Theorem \ref{thm-2square}.  
\end{remark}

\begin{example} 
Let $(m, q)=(2,3)$ and $\alpha$ be a generator of $\gf(r)^*$ with $\alpha^2 +2\alpha+2=0$. Then 
$\C_s$ is a $[8, 3, 5]$ ternary code with generator polynomial  
$$ 
\m_s(x)=x^5 + 2x^3 + x^2 + x + 1.    
$$
This cyclic code is an optimal linear code.
\end{example} 

\begin{example} 
Let $(m, q)=(3,3)$ and $\alpha$ be a generator of $\gf(r)^*$ with $\alpha^3 +2\alpha+1=0$. Then 
$\C_s$ is a $[26, 20, 4]$ ternary code with generator polynomial  
$$ 
\m_s(x)=x^6 + x^5 + x^3 + 2x + 2.    
$$
This cyclic code is an optimal linear code.
\end{example} 

\begin{example} 
Let $(m, q)=(4,3)$ and $\alpha$ be a generator of $\gf(r)^*$ with $\alpha^4 +2\alpha^3+2=0$. Then 
$\C_s$ is a $[80, 71, 5]$ ternary code with generator polynomial  
$$ 
\m_s(x)=  x^9 + 2x^8 + x^7 + 2x^6 + x^4 + x^2 + 1.    
$$
This cyclic code is an optimal linear code. 
\end{example} 

\begin{example} 
Let $(m, q)=(2,5)$ and $\alpha$ be a generator of $\gf(r)^*$ with $\alpha^2 +4\alpha+2=0$. Then 
$\C_s$ is a $[24, 19, 4]$ cyclic code over $\gf(5)$ with generator polynomial  
$$ 
\m_s(x)=x^5 + 3x^4 + 2x^3 + 3x^2 + 3x + 3.    
$$
This cyclic code is an optimal linear code. 
\end{example} 

\begin{example} 
Let $(m, q)=(3,5)$ and $\alpha$ be a generator of $\gf(r)^*$ with $\alpha^3 +3\alpha+3=0$. Then 
$\C_s$ is a $[124, 117, 4]$ cyclic code over $\gf(5)$ with generator polynomial  
$$ 
\m_s(x)=x^7 + 4x^6 + 4x^4 + 3x^2 + 3.    
$$
This cyclic code is an optimal linear code. 
\end{example}

\subsection{Cyclic codes from the Dembowski-Ostrom planar function}\label{sec-DOGold} 

The Dembowski-Ostrom planar function is given by $f(x)=x^{q^\kappa+1}$, where 
$m/\gcd(m,\kappa)$ and $q$ are odd. In this subsection we study the code $\C_s$ 
defined by this planar function. To this end, we need to prove the following lemma.  

\begin{lemma}\label{lem-DOGold} 
Let $m$ be odd. 
Let $s^{\infty}$ be the sequence of (\ref{eqn-sequence}), where $f(x)$ is the Dembowski-Ostrom planar function. 
Then the linear span $\ls_s$ of $s^{\infty}$ is equal to $2m+\N_p(m)$ and the minimal polynomial $\m_s$ 
of  $s^{\infty}$ are given by 
\begin{equation}\label{eqn-DOGold}
\m_s(x)= (x-1)^{\N_p(m)} m_{\alpha^{-1}}(x) m_{\alpha^{-(p^\kappa+1)}}(x)
\end{equation} 
where $m_{\alpha^{-j}}(x)$ is the minimal polynomial of $\alpha^{-j}$ over $\gf(q)$, 
$\N_p(i) =0$ if $i \equiv 0 \pmod{p}$ and $\N_p(i) =1$ otherwise. 
\end{lemma} 

\begin{proof} 
It is easily seen that 
\begin{eqnarray}\label{eqn-DOGold2}
s_t &=& \tr((\alpha^{t}+1)^{q^\kappa+1}) \nonumber \\ 
&=& \tr(1)+  2\sum_{j=0}^{m-1} (\alpha^t)^{q^j} + \sum_{j=0}^{m-1} (\alpha^t)^{(q^\kappa+1)q^j}. 
\end{eqnarray}

Since $m/\gcd(m,\kappa)$ is odd,  $\gcd(2\kappa, m)=\gcd(\kappa, m)$. 
Because $q$ is odd, $\gcd(q^\kappa-1, q^\kappa +1)=1$. 
It then follows that 
\begin{eqnarray*} 
\gcd(q^\kappa+1, q^m-1) 
&=&\frac{\gcd(q^{2\kappa}-1, q^m-1)}{\gcd(q^{\kappa}-1, q^m-1)} \\
&=& \frac{q^{\gcd(2\kappa, m)}-1}{q^{\gcd(\kappa, m)}-1} \\
&=& 1. 
\end{eqnarray*} 
Therefore, the size of the q-cyclotomic coset contating $q^\kappa+1$ is $m$. 
Clearly, $C_1 \cap C_{q^\kappa+1}=\emptyset$. 
The desired conclusions on the linear span and the minimal polynomial $\m_s(x)$ then follow from Lemma \ref{lem-ls2} 
and (\ref{eqn-DOGold2}). 
\end{proof}

The following theorem provides information on the code $\C_{s}$.    

\begin{theorem}\label{thm-DO2Gold} 
The code $\C_{s}$ defined by the sequence of Lemma \ref{lem-DOGold} has parameters 
$[n, n-2m-\N_p(m), d]$ and generator polynomial $\m_s(x)$ of (\ref{eqn-DOGold}), where 
\begin{eqnarray*} 
\left\{ \begin{array}{ll}
d =4             & \mbox{ if $q=3$ and $m \equiv 0 \pmod{p}$,}  \\
4 \le d \le 5  & \mbox{ if $q=3$ and $m \not\equiv 0 \pmod{p}$,} \\ 
d=3              & \mbox{ if $q >3$ and $m \equiv 0 \pmod{p}$,} \\ 
3 \le d  \le 4 & \mbox{ if $q >3$ and $m \not\equiv 0 \pmod{p}$.} 
\end{array} 
\right. 
\end{eqnarray*}   
\end{theorem} 

\begin{proof} 
The dimension of $\C_{s}$ follows from Lemma \ref{lem-DOGold} and the definition of the 
code $\C_s$.  We need to prove the conclusion on the minimum distance $d$ of $\C_{s}$. 
The code $\C_s$ of this theorem has the same weight distribution as the dual code in Theorem 
15 in \cite{CDY05} when $m \equiv 0 \pmod{p}$ and has the same weight distribution as the 
even-like subcode of the dual code in Theorem 15 in \cite{CDY05} when $m \not\equiv 0 \pmod{p}$.  
This is because the generator polynomial of the code $\C_s$ of this theorem and that of the  
dual code in Theorem 15 in \cite{CDY05} are reciprocals of each other when $m \equiv 0 \pmod{p}$, 
and the generator polynomial of the code $\C_s$ of this theorem and that of the even-like subcode 
of the 
dual code in Theorem 15 in \cite{CDY05} are reciprocals of each other when $m \not\equiv 0 \pmod{p}$.  
The conclusions on $d$ then follow from those on the minimum weight of the dual code in 
Theorem 15 in \cite{CDY05}.  
\end{proof} 

\begin{remark} 
Although the approach of this paper and that of \cite{CDY05} to the construction of linear 
codes  with planar functions are different, for the Dembowski-Ostrom planar function, the 
codes obtained with the two approaches are closely related. The relation between the two 
codes is made clear in the proof of  Theorem \ref{thm-DO2Gold}.  
\end{remark}

\begin{example} 
Let $(m,\kappa, q)=(3,1, 3)$ and $\alpha$ be a generator of $\gf(r)^*$ with $\alpha^3 + 2\alpha + 1=0$. Then 
$\C_s$ is a $[26, 20, 4]$ ternary code with generator polynomial  
$$ 
\m_s(x)= x^6 + 2x^5 + 2x^4+ x^3 + x^2 + 2x+ 1.    
$$
This cyclic code is an optimal linear code. 
\end{example}

\begin{example} 
Let $(m,\kappa, q)=(4,4, 3)$ and $\alpha$ be a generator of $\gf(r)^*$ with $\alpha^4 + 2\alpha^3 + 2=0$. Then 
$\C_s$ is a $[80, 71, 5]$ ternary code with generator polynomial  
$$ 
\m_s(x)= x^9 + 2x^8 + x^7 + 2x^6 + x^4 + x^2 + 1.    
$$
This cyclic code is an optimal linear code. 
\end{example}

\subsection{Cyclic codes from the planar functions $x^{10}-ux^6-u^2x^2$ over $\gf(3^m)$}\label{sec-DY} 

Throughout this subsection, let $q=3$ and let $m$ be odd. A family of planar functions 
$f(x)=x^{10}-ux^6-u^2x^2$ on $\gf(r)$ was discovered in \cite{CM,DY06}, where $u \in \gf(r)$. 
In this subsection we study the code $\C_s$ defined by these planar functions. To this end, we need to 
prove the following lemma.  

\begin{lemma}\label{lem-DY} 
Let $s^{\infty}$ be the sequence of (\ref{eqn-sequence}), where $f(x)=x^{10}-ux^6-u^2x^2$. 
Then the linear span $\ls_s$ of $s^{\infty}$ is given by 
\begin{eqnarray*}
\ls_s=\left\{ \begin{array}{l}
2m+\delta_u \mbox{ if } u^6+u=0 \\
3m+\delta_u \mbox{ otherwise,}  
\end{array}
\right. 
\end{eqnarray*}
and the minimal polynomial $\m_s(x)$ 
of  $s^{\infty}$ is given by 
\begin{eqnarray*}
\m_s(x)= \left\{ \begin{array}{l} 
(x-1)^{\delta_u} m_{\alpha^{-1}}(x) m_{\alpha^{-10}}(x)  \mbox{ if } u^6+u=0 \\
(x-1)^{\delta_u} m_{\alpha^{-1}}(x) m_{\alpha^{-2}}(x) m_{\alpha^{-10}}(x)  \mbox{ otherwise,}  
\end{array}
\right.  
\end{eqnarray*} 
where $m_{\alpha^{-j}}(x)$ is the minimal polynomial of $\alpha^{-j}$ over $\gf(q)$, 
$\delta_u=0$ if $\tr(u^2+u-1)=0$ and $\delta_u =1$ otherwise. 
\end{lemma} 

\begin{proof} 
By definition, we have 
$$ 
f(x+1)=x^{10}+x^9-ux^6-u^2x^2 + (1+u+u^2)x+(1-u-u^2). 
$$
It then follows that 
\begin{eqnarray*}
\lefteqn{\tr(f(x+1))=} \\
&\tr\left(x^{10}-(u^{3^{m-1}}+u^2)x^2 +(u^2+u-1)x\right)-\tr(u^2+u-1). 
\end{eqnarray*}
By definition, 
\begin{eqnarray}\label{eqn-DY2Gold2}
s_t = \tr\left((\alpha^t)^{10}-(u^{3^{m-1}}+u^2)(\alpha^t)^2 +v\alpha^t\right)-\tr(v),  
\end{eqnarray}
where $v=u^2+u-1$. 

It can be easily proved that 
$$
\ell_1=\ell_{n-1}=\ell_{2}=\ell_{n-2}=\ell_{10}=\ell_{n-10}=m 
$$ 
and that the three $q$-cyclotomic cosets $C_1, C_{2}$ and $C_{10}$ are pairwise disjoint. 

We now prove that $v =u^2+u-1\ne 0$ for all $u \in \gf(r)$. Suppose 
\begin{eqnarray}\label{eqn-351}
u^2+u-1=0. 
\end{eqnarray} 
Multiplying both sides of (\ref{eqn-351}) by $u$ yields 
\begin{eqnarray}\label{eqn-352}
u^3+u^2-u=0. 
\end{eqnarray} 
Combining (\ref{eqn-351}) and (\ref{eqn-352}) gives 
\begin{eqnarray}\label{eqn-353}
u^3+u+1=(u-1)^3 +(u-1)=0. 
\end{eqnarray}  
Since $m$ is odd, $-1$ is not a quadratic residue in $\gf(r)$. The only solution of $y^3+y=0$ 
is $y=0$. Hence the only solution of (\ref{eqn-353}) is $u=1$. However, $u=1$ is not a solution 
of (\ref{eqn-351}). Hence $v \ne 0$ for all $u$. 

Finally $u^{3^{m-1}}+u^2=0$ if and only if $u^6+u=0$.  
The desired conclusions on the linear span and the minimal polynomial $\m_s(x)$ then follow from Lemma \ref{lem-ls2},  
Equation (\ref{eqn-DY2Gold2}) and the conclusions on the cyclotomic cosets and their lengths. 
\end{proof}

The following theorem provides information on the code $\C_{s}$.    

\begin{theorem}\label{thm-DYY} 
The code $\C_{s}$ defined by the sequence of Lemma \ref{lem-DY} has parameters 
$[n, n-\ls_s, d]$ and generator polynomial $\m_s(x)$, where 
$\ls_s$ and $\m_s(x)$ are given in  Lemma \ref{lem-DY},  and 
\begin{eqnarray*} 
\left\{ \begin{array}{l}
5 \le d \le 8 \mbox{ if } u^6+u \ne 0 \mbox{ and } \delta_u=1; \\
4 \le d \le 6 \mbox{ if } u^6+u \ne 0 \mbox{ and } \delta_u=0; \\
3 \le d \le 6 \mbox{ if } u^6+u = 0 \mbox{ and } \delta_u=1; \\
3 \le d \le 4 \mbox{ if } u^6+u = 0 \mbox{ and } \delta_u=0. 
\end{array} 
\right. 
\end{eqnarray*} 
\end{theorem} 

\begin{proof} 
The dimension of $\C_{s}$ follows from Lemma \ref{lem-DY} and the definition of the 
code $\C_s$.  We need to prove the conclusions on the minimum distance $d$ of $\C_{s}$. 
It is known that the codes generated by any polynomial $g(x)$ and its reciprocal have the same 
weight distribution if $g(0) \ne 0$. 

When $u^6+u =0$, the reciprocal of $M_s(x)$ has root $\alpha^9$ and $\alpha^{10}$. In this case 
$d \ge 3$ by the BCH bound. The upper bounds on $d$ follow from the sphere-packing bound.   

When $u^6+u \ne 0$, the reciprocal of $M_s(x)$ has roots $\alpha^i$ for $i \in \{1,2,3\}$ 
and for all $i \in \{0,1,2,3\}$ if $\delta_u=1$.  The lower bounds on $d$ then come from 
the BCH bound. The upper bounds on $d$ follow from the sphere-packing bound. 
\end{proof} 

\begin{remark} 
The Coulter-Mathews planar function $x^{10}+x^6-x^2$ was employed in \cite{CDY05} 
to construct linear codes, which are not cyclic. However, the codes of this subsection are cyclic. 
On the other hand, the dimension of the codes of this subsection has two different possible 
values, as more planar functions are employed in this subsection.  This demonstrates that the 
two construction approaches of this paper and \cite{CDY05} are different. 

It is interesting to note that the dimensions of the codes defined by $x^{10}+x^6-x^2$ 
and $x^{10}-x^6-x^2$ are different though both are quadratic planar trinomials.    
\end{remark}

\begin{example} 
Let $(m, q, u)=(3,3,1)$ and $\alpha$ be a generator of $\gf(r)^*$ with $\alpha^3 +2\alpha+1=0$. Then 
$\C_s$ is a $[26, 17, 5]$ ternary code with generator polynomial  
$$ 
\m_s(x)=x^9 + x^8 + 2x^7 + 2x^6 + 2x^5 + x^4 + x^3 + x^2 + 2x + 1.    
$$
This cyclic code is an optimal linear code.
\end{example} 

\begin{example} 
Let $(m, q, u)=(3,3,-1)$ and $\alpha$ be a generator of $\gf(r)^*$ with $\alpha^3 +2\alpha+1=0$. Then 
$\C_s$ is a $[26, 20, 4]$ ternary code with generator polynomial  
$$ 
\m_s(x)=x^6 + 2x^5 + 2x^4 + x^3 + x^2 + 2x + 2.    
$$
This cyclic code is an optimal linear code.
\end{example} 

\begin{example} 
Let $(m, q, u)=(3,3,\alpha)$ and $\alpha$ be a generator of $\gf(r)^*$ with $\alpha^3 +2\alpha+1=0$. Then 
$\C_s$ is a $[26, 16, 6]$ ternary code with generator polynomial  
$$ 
\m_s(x)=x^{10} + x^8 + 2x^5 + x^2 + 2x + 2.    
$$
This cyclic code is an optimal linear code.
\end{example}

\subsection{Cyclic codes from $f(x)=x^{(q^h-1)/(q-1)}$}\label{sec-qhminus1} 

Let $h$ be a positive integer satistying the following condition: 
\begin{eqnarray}\label{eqn-qm1cond} 
\begin{array}{l} 
1 \le h \le \left\{ \begin{array}{l} 
                      (m-1)/2 \mbox{ if $m$ is odd and} \\
                      m/2 \mbox{ if $m$ is even.}   
                           \end{array} 
                           \right. 
\end{array} 
\end{eqnarray}

In this subsection, we deal with the code $\C_s$ defined by the sequence $s^{\infty}$ of 
(\ref{eqn-sequence}), where $f(x)=x^{(q^h-1)/(q-1)}$. When $h=1$, the code $\C_s$ 
has parameters $[n,n-m-\N_p(m), d]$, where $d=3$ if $\N_p(m)=1$ and $d=2$ if $\N_p(m)=0$. 
When $h=2$, the code 
$\C_s$ become a special case of the code in Section \ref{sec-DOGold}.  Therefore, we 
assume that $h \ge 3$ in this subsection.   

In order to study the code $\C_s$ of this subsection, we need to prove 
a number of auxiliary results on 
$q$-cyclotomic cosets. 

\begin{lemma}\label{lem-f261} 
Let $h$ satisfy the condition of (\ref{eqn-qm1cond}).  
For any $(i_1, i_2, \cdots, i_t)$ with $0 < i_1 < i_2 < \cdots < i_t \le h-1$, 
the size $\ell_i=|C_j| =\ell_{n-i}=m$, where $i=q^0+\sum_{j=1}^t q^{i_j}$. 
\end{lemma} 

\begin{proof} 
We prove the conclusion of this lemma only for the case that $m$ is even. The conclusion for 
$m$ being odd can be similarly proved. 

Let $u$ be any integer with $m-1 \ge u \ge 1$. 
Define 
$$ 
\Delta_1 = i(q^{u}-1), \ \Delta_2 = i(q^{m-u}-1).  
$$ 
Clearly $\Delta_i  \ne 0$ for both $i$ as $m-1 \ge u \ge 1$.

If $u \le m/2$, we have 
\begin{eqnarray*}
\Delta_1 &=& i(q^{u}-1) \\
& \le & (q^0+q^{\frac{m-2t}{2}} + q^{\frac{m-2t+2}{2}} + \cdots + q^{\frac{m-2}{2}}) (q^{\frac{m}{2}}-1) \\
& < & n. 
\end{eqnarray*} 
On the other hand, we have obviously that $\Delta_1>-i > -n$. Hence we have 
$\Delta_1 \not\equiv 0 \pmod{n}$ when  $u \le m/2$. 

 If $u > m/2$, then $m-u < m/2$. In this case one can similarly prove that 
 $-n < \Delta_2 < n$. So we have 
$\Delta_2 \not\equiv 0 \pmod{n}$ when  $u > m/2$. 

Summarizing the conclusions above proves the desired conclusions for $m$ being even. 
\end{proof} 

\begin{lemma}\label{lem-f262} 
Let $h$ satisfy the condition of (\ref{eqn-qm1cond}).  
For any pair of distinct $i=q^0+\sum_{l=1}^t q^{i_l}$ and $j=q^0+\sum_{l=1}^t q^{j_l}$ with 
$$ 
0 < i_1<i_2< \cdots i_t \le h-1 \mbox{ and } 0 < j_1<j_2< \cdots j_t \le h-1, 
$$ 
$C_i \cap C_j = \emptyset$, i.e., $i$ and $j$ cannot be in the same $q$-cyclotomic 
coset modulo $n$.  
\end{lemma} 

\begin{proof} 
Let $u$ be any integer with $0 \le u \le m -1$. 
Define 
$$
\Delta_1=i q^u -j \mbox{ and } \Delta_2=j q^{m-u} -i.  
$$ 

Notice that $i \equiv 1 \pmod{q}$ and $j \equiv 1 \pmod{q}$. We have  
that $ \Delta_i \ne 0$ for both $i$ as $i \ne j$. 

Suppose that $i$ and $j$ were in the same cyclotomic coset. Then $n$ would divide both 
$\Delta_1$ and $\Delta_2$. 

We prove the desired conclusion only for the case that $m$ is odd. The conclusions 
for $m$ being even can be similarly proved. 

We distinguish between the following two cases. When $u \le (m-1)/2$, we have 
\begin{eqnarray*}
\Delta_1 &=& iq^{u}-j \\
& \le & (q^0+q^{\frac{m-2t+1}{2}} + q^{\frac{m-2t+3}{2}} + \cdots + q^{\frac{m-1}{2}}) q^{\frac{m-1}{2}}-j \\
& < & n. 
\end{eqnarray*} 
On the other hand, we have obviously that $\Delta_1>-j> -n$. Hence we have 
$\Delta_1 \not\equiv 0 \pmod{n}$ when  $u \le (m-1)/2$. 

When $u > (m-1)/2$, we have $m-u \le (m-1)/2$. In this case one can similarly prove that 
$  
-n< \Delta_2 < n. 
$  
In this case $\Delta_2 \not\equiv 0 \pmod{n}$. 

Summarizing the conclusions of the two cases above proves the desired conclusion 
for $m$ being odd. 
\end{proof}

We need to do more preparations before presenting and proving the main results 
of this subsection. 

Let $J \ge t \ge 2$, and let $\N(J, t)$ denote the total number of vectors 
$(i_1, i_2, \cdots, i_{t-1})$ wth $1 \le i_1<i_2<\cdots < i_{t-1}<J$. 
By definition, we have the following recursive formula: 
\begin{eqnarray}\label{eqn-recur}
\N(J, t)=\sum_{j=t-1}^{J-1} \N(j, t-1). 
\end{eqnarray} 
By definition, we have 
\begin{eqnarray}\label{eqn-J2}
\N(J,2)=J-1 \mbox{ for all } J \ge 2  
\end{eqnarray} 
and 
\begin{eqnarray}\label{eqn-J3}
\N(J,3)= \frac{(J-1)(J-2)}{2} \mbox{ for all } J \ge 3.   
\end{eqnarray} 
It the follows from (\ref{eqn-recur}),  (\ref{eqn-J2}) and (\ref{eqn-J3}) that 
\begin{eqnarray}\label{eqn-J4}
\N(J, 4) 
&=& \sum_{j=3}^{J-1} \N(j, 3) \nonumber \\
&=& \sum_{j=3}^{J-1} \frac{(J-1)(J-2)}{2}  \nonumber \\ 
&=& \frac{J^3-6J^2+11J-6}{6}. 
\end{eqnarray} 
By definition, we have 
\begin{eqnarray}\label{eqn-JJ}
\N(t,t)=1 \mbox{ for all } t \ge 2.   
\end{eqnarray} 
For convenience, we define $\N(J,1)=1$ for all $J \ge 1$.

\begin{lemma}\label{lem-2mm1} 
Let $h$ satisfy the condition of (\ref{eqn-qm1cond}).  
Let $s^{\infty}$ be the sequence of (\ref{eqn-sequence}), where $f(x)=x^{(q^h-1)/(q-1)}$. Then the linear 
span $\ls_s$ and minimal polynomial $\m_s(x)$ of $s^{\infty}$ are given by 
\begin{eqnarray*}\label{eqn-2m0} 
\ls_s = \left(\N_p(h)+ \sum_{t=1}^{h-1} \sum_{u=1}^{h-1} \N_p(h-u) \N(u, t) \right) m + \N_p(m)  
\end{eqnarray*} 
and 
\begin{eqnarray*}\label{eqn-2m21}
\m_s(x) 
&=& (x-1)^{\N_p(m)}  m_{\alpha^{-1}}(x)^{\N_p(h)}  \prod_{1 \le u \le h-1 \atop \N_p(h-u)=1} m_{\alpha^{-(q^0+q^u)}}(x)  \nonumber \\
&\times & \prod_{t=2}^{h-1}   \prod_{t \le u \le h-1 \atop \N_p(h-u)=1} 
\prod_{1 \le i_1<\cdots < i_{t-1} < u} m_{\alpha^{-(q^0+\sum_{j=1}^{t-1} q^{i_j}+ q^u)}}(x)   \nonumber   \\          
\end{eqnarray*} 
\end{lemma} 

\begin{proof} 
Define 
$$
x=\alpha^t. 
$$  
Then we have 
\begin{eqnarray}\label{eqn-2m11}
s_t  
&=& \tr\left(  (x+1)^{\sum_{i=0}^{h-1} q^{i}}    \right) \nonumber \\
&=& \tr\left( \prod_{i=0}^{h-1} \left(x^{q^{i}}+1\right)     \right) \nonumber \\
&=& \tr(1) + \tr\left[ \sum_{t=1}^{h}  \sum_{0 \le i_1 < \cdots i_t \le h-1} x^{\sum_{j=1}^t q^{i_j} }     \right] \nonumber \\
&=& \tr(1) + h\tr(x) + \tr\left[ \sum_{i_1=1}^{h-1} (h-i_1)  x^{q^0+q^{i_1}}     \right] + \nonumber \\
& & 
\tr\left[ \sum_{t=2}^{h-1}  \sum_{i_1=1}^{h-t} \sum_{i_2=i_1+1}^{h-t+1} \cdots  \sum_{i_t=i_{t-1}+1}^{h-1} (h-i_t) 
  x^{q^0 + \sum_{j=1}^t q^{i_j} }     \right] \nonumber \\
&=& \tr(1) + h\tr(x) + \tr\left[ \sum_{i_1=1}^{h-1} (h-i_1)  x^{q^0+q^{i_1}}     \right] + \nonumber \\
& & 
\tr\left[ \sum_{t=2}^{h-1} \sum_{i_t=t}^{h-1} (h-i_t) \sum_{1 \le i_1 < \cdots < i_{t-1} < i_t} x^{q^0+\sum_{j=1}^t q^{i_j} }     \right]. 
\end{eqnarray} 

The desired conclusions on the linear span and the minimal polynomial $\m_s(x)$ then follow from Lemmas \ref{lem-ls2}, 
\ref{lem-f261}, \ref{lem-f262},  and Equation   
(\ref{eqn-2m11}). 
\end{proof}

The following theorem provides information on the code $\C_{s}$.    

\begin{theorem}\label{thm-qh} 
The code $\C_{s}$ defined by the sequence of Lemma \ref{lem-2mm1} has parameters 
$[n, n-\ls_s, d]$ and generator polynomial $\m_s(x)$, 
 where $\ls_s$ and $\m_s(x)$ are given in Lemma \ref{lem-2mm1}. 
\end{theorem} 

\begin{proof} 
The dimension of $\C_{s}$ follows from Lemma \ref{lem-2mm1} and the definition of the 
code $\C_s$. 
\end{proof} 

As a corollary of Theorem \ref{thm-qh}, we have the following.  
\begin{corollary}\label{cor-qh1} 
Let $h =3$. 
The code $\C_{s}$ of Theorem \ref{thm-qh} has parameters  
$[n, n-\ls_s, d]$ and generator polynomial $\m_s(x)$ given by 
\begin{eqnarray*}
\m_s(x) = 
(x-1)^{\N_p(m)} m_{\alpha^{-1}}(x) m_{\alpha^{-1-q}}(x) m_{\alpha^{-1-q^2}}(x)  m_{\alpha^{-1-q-q^2}}(x)   
\end{eqnarray*}
if $p \ne 3$; and 
\begin{eqnarray*}
\m_s(x) = 
(x-1)^{\N_p(m)} m_{\alpha^{-1-q}}(x) m_{\alpha^{-1-q^2}}(x)  m_{\alpha^{-1-q-q^2}}(x)   
\end{eqnarray*}
if $p=3$, where 
\begin{eqnarray}
\ls_s=\left\{ \begin{array}{l}
4m+\N_p(m) \mbox{ if } p \ne 3 \\
3m+\N_p(m) \mbox{ if } p = 3. 
\end{array}
\right. 
\end{eqnarray}
In addition,  
\begin{eqnarray*} 
\left\{ \begin{array}{ll}
3 \le d \le 8  & \mbox{if $p=3$ and $\N_p(m)=1$} \\
3 \le d \le 6  & \mbox{if $p=3$ and $\N_p(m)=0$} \\
3 \le d \le 8  & \mbox{if $p>3$.}  
\end{array}
\right. 
\end{eqnarray*}
\end{corollary} 

\begin{proof} 
We need to prove only the bounds on the minimum weight of this code. 
The upper bounds on $d$ follow from the sphere-packing bound and the dimension of the code. 
In both cases, the reciprocal of $\m_s(x)$ has the roots $\alpha^{q+q^2}$ and  $\alpha^{1+q+q^2}$. 
It then follows from the BCH bound that $d \ge 3$. 
\end{proof}

\begin{open} 
For the code $\C_s$ of Corollary \ref{cor-qh1}, do the following lower bounds hold?  
\begin{eqnarray*} 
d \ge \left\{ \begin{array}{ll}
5   & \mbox{when $p=3$ and $\N_p(m)=1$} \\
4   & \mbox{when $p=3$ and $\N_p(m)=0$} \\ 
6   & \mbox{when $p>3$ and $\N_p(m)=1$} \\
5   & \mbox{when $p>3$ and $\N_p(m)=0$.} 
\end{array}
\right. 
\end{eqnarray*}
\end{open}

\begin{example} 
Let $(m,h, q)=(2,3,3)$ and $\alpha$ be a generator of $\gf(r)^*$ with $\alpha^2 +2\alpha+2=0$. Then 
$\C_s$ is a $[8, 2, 6]$ ternary code with generator polynomial  
$$ 
\m_s(x)=x^6 + 2x^5 + 2x^4 + 2x^2 + x + 1.    
$$
This cyclic code is optimal. Notice that $h >m/2$. Hence, the parameters of this code do not 
agree with those of the code in Corollary \ref{cor-qh1}. 
In this case $f(x)$ is a 
permutation.  
\end{example} 

\begin{example} 
Let $(m,h, q)=(3,3,3)$ and $\alpha$ be a generator of $\gf(r)^*$ with $\alpha^3 +2\alpha+1=0$. Then 
$\C_s$ is a $[26, 26, 1]$ ternary code with generator polynomial  
$  
\m_s(x)=1.    
$
This cyclic code is optimal as it is MDS. Notice that $h >(m-1)/2$. Hence, the parameters of this code do not 
agree with those of the code in Corollary \ref{cor-qh1}. 
In this case $f(x)$ is not a 
permutation. In fact, $\gcd((q^h-1)/(q-1), q^m-1)=13$.   
\end{example} 

\begin{example} 
Let $(m,h, q)=(4,3,3)$ and $\alpha$ be a generator of $\gf(r)^*$ with $\alpha^4 +2\alpha^3+2=0$. Then 
$\C_s$ is a $[80, 69, 5]$ ternary code with generator polynomial  
$$ 
\m_s(x)= x^{11} + 2x^8 + 2x^6 + 2x^5 + 2x^4 + x^3 + 2x^2 + x + 2.    
$$
This is an almost optimal linear code. The known optimal linear code has parameters $[80, 69, 6]$ 
which is not cyclic. Notice that $h >m/2$. Hence, the parameters of this code do not 
agree with those of the code in Corollary \ref{cor-qh1}. 
In this case $f(x)$ is a 
permutation.  
\end{example} 

\begin{example} 
Let $(m,h, q)=(5,3,3)$ and $\alpha$ be a generator of $\gf(r)^*$ with $\alpha^5 +2\alpha+1=0$. Then 
$\C_s$ is a $[242, 226, d]$ ternary code with generator polynomial  
$$ 
\m_s(x)= x^{16} + 2x^{14} + 2x^{12} + 2x^{11} + x^{10} + x^9 + x^6 + x^3 + 2x^2 + 2.    
$$
Notice that $h >m/2$. However, the parameters of this code do  
agree with those of the code in Corollary \ref{cor-qh1}. 
In this case $f(x)$ is a 
permutation.  
\end{example} 

\begin{example} 
Let $(m,h, q)=(6,3,3)$ and $\alpha$ be a generator of $\gf(r)^*$ with $\alpha^6 + 2\alpha^4 + \alpha^2 + 2\alpha + 2=0$. Then 
$\C_s$ is a $[728, 710, d]$ ternary code with generator polynomial  
\begin{eqnarray*}
\m_s(x) &=& x^{18} + 2x^{15} + 2x^{14} + 2x^{13} + 2x^{11} + x^{10} \\
& & + 2x^9 + x^8 + x^6 + 2x^4 + x^3 + x^2 + 2.    
\end{eqnarray*} 
Notice that $h = m/2$. Hence, the parameters of this code  
agree with those of the code in Corollary \ref{cor-qh1}. 
In this case $f(x)$ is not a 
permutation. In fact, $\gcd((q^h-1)/(q-1), q^m-1)=13$.  
\end{example}

\begin{example} 
Let $(m,h, q)=(2,3,5)$ and $\alpha$ be a generator of $\gf(r)^*$ with $\alpha^2 +4\alpha+2=0$. Then 
$\C_s$ is a $[24, 16, 5]$ cyclic code over $\gf(5)$ with generator polynomial  
$$ 
\m_s(x)= x^8 + x^7 + 2x^4 + 2x^3 + 3x^2 + 4x + 2.    
$$
The best linear code known has parameters $[24, 16, 6]$ 
which is not cyclic.
Notice that $h >m/2$. Hence, the parameters of this code do not 
agree with those of the code in Corollary \ref{cor-qh1}. 
%In this case $f(x)$ is a 
%permutation.  
\end{example}

\begin{example} 
Let $(m,h, q)=(6,3,5)$ and $\alpha$ be a generator of $\gf(r)^*$ with $\alpha^6 + \alpha^4 + 4\alpha^3 + \alpha^2 + 2=0$. Then 
$\C_s$ is a $[15624, 15599, d]$ cyclic code over $\gf(5)$ with generator polynomial  
\begin{eqnarray*}
\m_s(x) &=&  x^{25} + x^{24} + 3x^{23} + 2x^{22} + 3x^{21} + x^{20} + 2x^{19} + \\ 
& & 4x^{18} + 4x^{17} + x^{16} + 2x^{14} + 4x^{12} + 2x^{11} + 3x^{10} + \\
& & 4x^8 + 4x^6 + 4x^5 + x^4 + 4x^3 + x^2 + 4.    
\end{eqnarray*} 
Notice that $h = m/2$. Hence, the parameters of this code   
agree with those of the code in Corollary \ref{cor-qh1}. 
%In this case $f(x)$ is not a 
%permutation. In fact, $\gcd((q^h-1)/(q-1), q^m-1)=31$.  
\end{example}

\subsection{Cyclic codes from the Coulter-Mathews planar function }\label{sec-CM} 

Throughout this subsection, let $q=3$, and $r=q^m$ as before. 
Let $h$ be a positive integer satistying the following conditions: 
\begin{eqnarray}\label{eqn-CM2m1cond} 
\left\{ 
\begin{array}{l}  
h \mbox{ is odd,} \\
\gcd(m, h)=1, \\
3 \le h \le \left\{ \begin{array}{l} 
                      (m-1)/2 \mbox{ if $m$ is odd and} \\
                      m/2 \mbox{ if $m$ is even.}   
                           \end{array} 
                           \right. 
\end{array} 
\right. 
\end{eqnarray}

In this subsection, we deal with the code $\C_s$ defined by the sequence $s^{\infty}$ of 
(\ref{eqn-sequence}), where $f(x)=x^{(3^h+1)/2}$. When $h=1$,  the code 
$\C_s$ becomes the code of  Section \ref{sec-2square}.  Therefore, we consider the 
case  
$h \ge 3$ in this subsection.  The code of this subsection is related to the code of Section 
\ref{sec-qhminus1}, so we need to use some of the notations, symbols and lemmas 
of Section \ref{sec-qhminus1}.   

In order to study the code $\C_s$ of this subsection, we need to prove 
a number of auxiliary results on 
$q$-cyclotomic cosets. 

\begin{lemma}\label{lem-CMf261} 
Let $h$ satisfy the third condition of (\ref{eqn-CM2m1cond}).  
For any $(i_1, i_2, \cdots, i_t)$ with $0 < i_1 < i_2 < \cdots i_t \le h-1$, 
the size $\ell_i=|C_j| =\ell_{n-i}=m$, where $i=2+\sum_{j=1}^t 3^{i_j}$. 
In addition, $\ell_2=\ell_{n-2}=m$.  
\end{lemma} 

\begin{proof} 
The proof of Lemma \ref{lem-f261} is easily modified into a proof of this lemma.  
We omit the details here. 
\end{proof} 

\begin{lemma}\label{lem-CMf262} 
Let $h$ satisfy the third condition of (\ref{eqn-CM2m1cond}).  
For any pair of distinct $i=2+\sum_{u=1}^t 3^{i_u}$ and $j=2+\sum_{u=1}^t 3^{j_u}$ with 
$$ 
0 < i_1<i_2< \cdots i_t \le h-1 \mbox{ and } 0 < j_1<j_2< \cdots j_t \le h-1, 
$$ 
$C_i \cap C_j = \emptyset$, i.e., $i$ and $j$ cannot be in the same $q$-cyclotomic 
coset modulo $n$.  In addition, 
\begin{itemize} 
\item $C_2 \cap C_1 = \emptyset.$ 
\item $C_2 \cap C_{1+\sum_{u=1}^t 3^{i_u}} = \emptyset$, where  $1 \le i_1 < \cdots < i_t \le h-1$.  
\item $C_2 \cap C_{2+\sum_{u=1}^t 3^{i_u}} = \emptyset$, where  $1 \le i_1 < \cdots < i_t \le h-1$. 
\item $C_{2+\sum_{u=1}^{t_1} 3^{i_u}}  \cap C_{1+\sum_{u=1}^{t_2} 3^{i_u}} = \emptyset$, where 
$$ 
1 \le i_1 < \cdots < i_{t_1} \le h-1 \mbox{ and } 1 \le i_1 < \cdots < i_{t_2} \le h-1.  
$$ 
\end{itemize} 
\end{lemma} 

\begin{proof} 
The proof of Lemma \ref{lem-f262} can be modified into a proof of this lemma. 
We omit the details here.  
\end{proof}

\begin{lemma}\label{lem-CM2mm1} 
Let $h$ satisfy the third condition of (\ref{eqn-CM2m1cond}).  
Let $s^{\infty}$ be the sequence of (\ref{eqn-sequence}), where $f(x)=x^{(3^h+1)/2}$. Then the linear 
span $\ls_s$ and minimal polynomial $\m_s(x)$ of $s^{\infty}$ are given by 
\begin{eqnarray*}\label{eqn-CM2m0} 
\ls_s &=& \N_3(m) + \left(\sum_{i=0}^h\N_3(h-i+1) \right)m + \\
        &  & \left( \sum_{t=2}^h \N(h,t) + \sum_{t=2}^{h-1} \sum_{i_t=t}^{h-1} \N_3(h-i_t+1) \N(i_t, t) \right) m   
\end{eqnarray*} 
and 
\begin{eqnarray*}\label{eqn-CM2m21}
\m_s(x) 
&=& (x-1)^{ \N_3(m)}  m_{\alpha^{-1}}(x)^{\N_3(h+1)} m_{\alpha^{-2}}(x) \times \\
& & \prod_{t=1}^{h-1} \prod_{1 \le i_1 < \cdots < i_t \le h-1} m_{\alpha^{-(2+\sum_{j=1}^t 3^{i_j})}}(x) \times \\  
& & \prod_{1 \le u \le h-1 \atop \N_3(h-u+1)=1} m_{\alpha^{-(1+3^u)}}(x)  \times \\
& & \prod_{t=2}^{h-1}   \prod_{t \le i_t \le h-1 \atop \N_3(h-i_t+1)=1}   
      \prod_{1 \le i_1<\cdots < i_{t-1} < i_t} m_{\alpha^{-(1+\sum_{j=1}^{t} 3^{i_j})}}(x),            
\end{eqnarray*} 
where $\N_3(j)$ and $\N(j, t)$ were defined in Sections \ref{sec-notations} and \ref{sec-qhminus1} respectively.  
\end{lemma} 

\begin{proof} 
Note that 
$$ 
\frac{3^h+1}{2}=1+\sum_{i=0}^{h-1} 3^i. 
$$
Define 
$
x=\alpha^t
$   
and 
$$ 
f_h(x) =\tr\left(  (x+1)^{\sum_{i=0}^{h-1} 3^{i}} \right). 
$$
Then we have 
\begin{eqnarray}\label{eqn-CM2m111}
s_t  
&=& \tr\left(  (x+1)^{1+\sum_{i=0}^{h-1} 3^{i}}    \right) \nonumber \\
&=& \tr\left(  x(x+1)^{\sum_{i=0}^{h-1} 3^{i} }    \right)  + f_h(x)\nonumber \\
&=& \tr(x) + \tr\left[ \sum_{t=1}^{h} \left[ \sum_{0 \le i_1 < \cdots < i_t \le h-1} x^{1+\sum_{j=1}^t 3^{i_j} } \right] \right] + f_h(x) \nonumber \\
&=& f_h(x) + \tr(x) + \tr\left[ \sum_{i_1=0}^{h-1}   x^{1+3^{i_1}}     \right] + \nonumber \\
& & \tr\left[ \sum_{t=2}^{h}  \left(\sum_{0 \le i_1 < \cdots < i_t \le h-1} x^{1+\sum_{j=1}^t 3^{i_j} }  \right)   \right] \nonumber \\
&=& f_h(x) + \tr(x) + \tr(x^2)+\tr\left[ \sum_{i_1=1}^{h-1}   x^{1+3^{i_1}}     \right] + \nonumber \\
& & \tr\left[ \sum_{t=2}^{h} \left( \sum_{1 \le i_2 < \cdots < i_t \le h-1} x^{2+\sum_{j=1}^t 3^{i_j} }    \right) \right] +  \nonumber \\
& & \tr\left[ \sum_{t=2}^{h-1} \left( \sum_{1 \le i_1 < \cdots < i_t \le h-1} x^{1+\sum_{j=1}^t 3^{i_j} }  \right)   \right] 
\end{eqnarray} 

Using the expression of (\ref{eqn-2m11}) for $f_h(x)$ and merging terms in (\ref{eqn-CM2m111}), we obtain 
\begin{eqnarray}\label{eqn-CM2m11}
\lefteqn{s_t=} \nonumber \\  
&& \tr(1) + (h+1)\tr(x) + \tr(x^2)+ \nonumber \\
&& \tr\left[ \sum_{i_1=1}^{h-1} (h-i_1+1)  x^{1+3^{i_1}}     \right] + \nonumber \\
&& \tr\left[ \sum_{t=2}^{h}  \left(\sum_{1 \le i_2 < \cdots < i_t \le h-1} x^{2+\sum_{j=1}^t 3^{i_j} } \right)   \right] +  \nonumber \\
&& \tr\left[ \sum_{t=2}^{h-1} \sum_{i_t=t}^{h-1} (h-i_t+1)\sum_{1 \le i_1 < \cdots < i_{t-1} \le i_t} x^{1+\sum_{j=1}^t 3^{i_j} }  \right] 
\end{eqnarray} 
 
The desired conclusions on the linear span and the minimal polynomial $\m_s(x)$ then follow from Lemmas \ref{lem-ls2}, 
\ref{lem-f261}, \ref{lem-CMf261}, \ref{lem-f262},  \ref{lem-CMf262},  and Equation   
(\ref{eqn-CM2m11}). 
\end{proof}

The following theorem provides information on the code $\C_{s}$.    

\begin{theorem}\label{thm-CMqh2} 
The code $\C_{s}$ defined by the sequence of Lemma \ref{lem-CM2mm1} has parameters 
$[n, n-\ls_s, d]$ and generator polynomial $\m_s(x)$, 
 where $\ls_s$ and $\m_s(x)$ are given in Lemma \ref{lem-CM2mm1}. 
\end{theorem} 

\begin{proof} 
The dimension of $\C_{s}$ follows from Lemma \ref{lem-CM2mm1} and the definition of the 
code $\C_s$. 
\end{proof} 

\begin{remark} 
The Coulter-Mathews planar function $f(x)=x^{(3^h+1)/2}$ was employed in \cite{CDY05} 
to construct cyclic codes whose dual codes have dimension $n-2m$, which is independent 
of $h$. However, the dimension of the code of Theorem \ref{thm-CMqh2} depends on $h$. 
This once again shows the difference of the construction approach of this paper and that of 
\cite{CDY05}.    
\end{remark} 

As a corollary of Theorem \ref{thm-CMqh2}, we have the following.  
\begin{corollary}\label{cor-CMqh1} 
Let $h =3$. 
The code $\C_{s}$ of Theorem \ref{thm-CMqh2} has parameters  
$[n, n-\ls_s, d]$ and the  generator polynomial $\m_s(x)$ given by 
\begin{eqnarray*}
\m_s(x) 
&=& (x-1)^{\N_3(m)} m_{\alpha^{-1}}(x) m_{\alpha^{-2}}(x)  m_{\alpha^{-5}}(x) \times \\
&  &  m_{\alpha^{-10}}(x) m_{\alpha^{-11}}(x)  m_{\alpha^{-13}}(x)  m_{\alpha^{-14}}(x)      
\end{eqnarray*}
where 
\begin{eqnarray*}
\ls_s=7m+\N_3(m). 
\end{eqnarray*}
In addition, 
\begin{eqnarray*} 
\left\{ \begin{array}{ll}
5 \le d \le 16  & \mbox{if $\N_3(m)=1$} \\
4 \le d \le 16  & \mbox{if $\N_3(m)=0$.} 
\end{array}
\right. 
\end{eqnarray*}
\end{corollary}

\begin{proof} 
We need to prove only the bounds on the minimum weight of this code. 
The upper bound on $d$ follows from the sphere-packing bound and the dimension of the code. 
In both cases, the reciprocal of $\m_s(x)$ has the roots $\alpha^{i}$ for all   $i \in \{1,2,3 \}$. 
When $\N_3(m)=1$,  the reciprocal of $\m_s(x)$ has the additional root $\alpha^0$. 
The lower bounds then follow from the BCH bound. 
\end{proof}

\begin{open} 
For the code $\C_s$ of Corollary \ref{cor-CMqh1}, do the following lower bounds hold?  
\begin{eqnarray*} 
d \ge \left\{ \begin{array}{ll}
9   & \mbox{when $\N_p(m)=1$} \\
8   & \mbox{when $\N_p(m)=0$.}  
\end{array}
\right. 
\end{eqnarray*}
\end{open}

\begin{example} 
Let $(m,h, q)=(2,3,3)$ and $\alpha$ be a generator of $\gf(r)^*$ with $\alpha^2 +2\alpha+2=0$. Then 
$\C_s$ is a $[8, 3, 5]$ ternary code with generator polynomial  
$$ 
\m_s(x)= x^5 + 2x^3 + x^2 + x + 1.    
$$
This cyclic code is optimal. Notice that $h >m/2$. Hence, the parameters of this code do not 
agree with those of the code in Corollary \ref{cor-CMqh1}. 
\end{example}

\begin{example} 
Let $(m,h, q)=(4,3,3)$ and $\alpha$ be a generator of $\gf(r)^*$ with $\alpha^4 +2\alpha^3+2=0$. Then 
$\C_s$ is a $[80, 69, 5]$ ternary code with generator polynomial  
$$ 
\m_s(x)= x^{11} + 2x^8 + 2x^6 + 2x^5 + 2x^4 + x^3 + 2x^2 + x + 2.    
$$
This is an almost optimal linear code. The known optimal linear code has parameters $[80, 69, 6]$ 
which is not cyclic. Notice that $h >m/2$. Hence, the parameters of this code do not 
agree with those of the code in Corollary \ref{cor-CMqh1}. In this case $f(x)$ is a 
permutation.  
\end{example}

\begin{example} 
Let $(m,h, q)=(7,3,5)$ and $\alpha$ be a generator of $\gf(r)^*$ with $\alpha^7 + 2 \alpha^2 + 1=0$. Then 
$\C_s$ is a $[2186, 2136, d]$  cyclic code over $\gf(5)$ with generator polynomial  
\begin{eqnarray*}
\m_s(x) &=& x^{50} + x^{49} + x^{48} + 2x^{47} + 2x^{46} + x^{45} + 2x^{44} +\\
             & &  2x^{43} + x^{42} + x^{41} + 2x^{40} + 2x^{39} + 2x^{38} + 2x^{37} + \\ 
             & &  x^{36} + 2x^{35} + 2x^{34} + 2x^{33} + x^{31} + 2x^{30} + x^{29} + \\ 
             & & 2x^{28} + 2x^{27} + 2x^{26} + 2x^{25} + x^{24} + x^{23} +x^{22} + \\
             & & 2x^{21} + 2x^{20} + x^{18} + x^{16} + x^{15} + x^{14} + x^{13} + \\ 
             & & 2x^{12} + x^{11} +2x^{10} + 2x^9 + 2x^4 + 1.   
\end{eqnarray*} 
Notice that $h = (m-1)/2$. Hence, the parameters of this code   
agree with those of the code in Corollary \ref{cor-CMqh1}. 
\end{example}

\subsection{Cyclic codes from the APN function $x^3$}\label{sec-x3} 

In this subsection we consider the code $\C_s$ defined by the APN function $f(x)=x^3$ over $\gf(r)$, where 
$p>3$. We present the results without providing proofs as this case is simple.

\begin{lemma}\label{lem-cubic} 
Let $s^{\infty}$ be the sequence of (\ref{eqn-sequence}), where $f(x)=x^3$ and $p>3$. Then the linear span $\ls_s$
of $s^{\infty}$ is equal to $3m+\N_p(m)$ and the minimal polynomial $\m_s(x)$ of  $s^{\infty}$ is given by 
\begin{equation}\label{eqn-gpcubic}
\m_s(x)=(x-1)^{\N_p(m)} m_{\alpha^{-1}}(x) m_{\alpha^{-2}}(x) m_{\alpha^{-3}}(x)
\end{equation} 
where $m_{\alpha^j}(x)$ is the minimal polynomial of $\alpha^j$ over $\gf(q)$. 
\end{lemma}

The following theorem provides information on the code $\C_{s}$ and its dual.    

\begin{theorem} 
The code $\C_{s}$ defined by the sequence of Lemma \ref{lem-cubic} has parameters 
$[n, n-3m-\N_p(m), d]$ and generator polynomial $\m_s(x)$ of (\ref{eqn-gpcubic}), where  
$4 \le d \le 8$. When $\N_p(m)=1$, $5 \le d \le 6$. 
\end{theorem} 

\begin{example} 
Let $(m,q)=(2,5)$ and $\alpha$ be a generator of $\gf(q^m)^*$ with $\alpha^2 + 4\alpha + 2=0$. 
Then generator polynomial of the code $\C_s$ is 
$$ 
\m_s(x)=x^7 + 3x^6 + 4x^5 + 4x^4 + 2x^3 + 4x^2 + x + 1   
$$
and $\C_s$ is a $[24, 17, 5]$ cyclic code over $\gf(5)$. The upper bound on the minimum weight 
of any linear code of length 24 and dimension 17 over $\gf(5)$ is 6. The record linear code with 
parameters $[24,17,6]$ reported in the database maintained by Markus Grassl is not cyclic.  
\end{example} 

\begin{example} 
Let $(m,q)=(3,5)$ and $\alpha$ be a generator of $\gf(q^m)^*$ with $\alpha^3 + 3\alpha + 3=0$. 
Then generator polynomial of the code $\C_s$ is 
$$ 
\m_s(x)=x^{10} + x^9 + x^5 + 3x^4 + 4x^3 + x + 4   
$$
and $\C_s$ is a $[124, 114, 5]$ cyclic code over $\gf(5)$. The upper bound on the minimum weight 
of any linear code of length 125 and dimension 114 over $\gf(5)$ is 6. The record linear code with 
parameters $[124,114,6]$ reported in the database maintained by Markus Grassl is not cyclic.  
\end{example}

\subsection{Open problems regarding the nonbinary cyclic codes from APN and planar dunctions}

In previous subsections of Section \ref{sec-nonbinary}, some nonbinary cyclic codes from APN and 
planar functions were studied.  It would be nice if the following open problems could be solved.   

\begin{open}\label{open-11} 
Determine the dimension and the generator polynomial of the code $\C_s$ defined by the APN function 
$f(x)=x^{(3^m-3)/2}$ over $\gf(3)$.  Develop 
tight lower bounds on the minimum weight of this code.   
\end{open}

The following examples demonstrate that the code described in Open Problem \ref{open-11} 
looks promising. 

\begin{example} 
Let $(m,q)=(3,3)$ and $\alpha$ be a generator of $\gf(q^m)^*$ with $\alpha^3 + 2\alpha + 1=0$. 
Then generator polynomial of the code $\C_s$ is 
$$ 
\m_s(x)=x^6 + 2x^5 + 2x^4 + x^3 + x^2 + 2x + 2   
$$
and $\C_s$ is a $[26, 20, 4]$ cyclic code over $\gf(3)$, and is an optimal linear code. 
\end{example} 

\begin{example} 
Let $(m,q)=(4,3)$ and $\alpha$ be a generator of $\gf(q^m)^*$ with $\alpha^4 + 2\alpha^3 + 2=0$. 
Then generator polynomial of the code $\C_s$ is 
$$ 
\m_s(x)=x^{11} + 2x^8 + 2x^6 + 2x^5 + 2x^4 + x^3 + 2x^2 + x + 2   
$$
and $\C_s$ is a $[80, 69, 5]$ cyclic code over $\gf(3)$. The upper bound on the minimum weight 
of any linear code of length 80 and dimension 69 over $\gf(3)$ is 6. The record linear code with 
parameters $[80,69,6]$ reported in the database maintained by Markus Grassl is not cyclic.  
\end{example} 

\begin{open} 
Determine the dimension and the generator polynomial of the code $\C_s$ defined by the inverse  
APN function $f(x)=x^{q^m-2}$ over $\gf(q)$, where $q$ is odd. Develop tight lower bounds on 
the minimum weight of this code.   
\end{open}  

Regarding the cyclic code defined by the inverse APN function, the case for $q=2$ was settled in 
Section \ref{sec-Inverse}. The case $q$ being odd is more complicated. The code is interesting in 
the binary case. 

\begin{open} 
Determine the dimension and the generator polynomial of the code $\C_s$ defined by the    
APN function $f(x)=x^{(5^h+1)/2}$ over $\gf(5^m)$, where $\gcd(2m,h)=1$.  Develop tight 
lower bounds on the minimum weight of this code.   
\end{open}

\begin{open} 
Determine the dimension and the generator polynomial of the ternary code $\C_s$ defined by 
the  APN function $f(x)=x^{e}$ over $\gf(3^m)$, where  
\begin{eqnarray*} 
e=\left\{ \begin{array}{ll} 
               \frac{3^{(m+1)/2}-1}{2} & \mbox{if } m \equiv 3 \pmod{4} \\
               \frac{3^{(m+1)/2}-1}{2} + \frac{3^m-1}{2} & \mbox{if } m \equiv 1 \pmod{4}.                 
               \end{array} 
\right.                
\end{eqnarray*}   
Develop 
tight lower bounds on the minimum weight of this code.   
\end{open}  

\begin{open} 
Determine the dimension and the generator polynomial of the ternary code $\C_s$ defined by 
the  APN function $f(x)=x^{e}$ over $\gf(3^m)$, where  
\begin{eqnarray*} 
e=\left\{ \begin{array}{ll} 
               \frac{3^{m+1}-1}{8} & \mbox{if } m \equiv 3 \pmod{4} \\
               \frac{3^{m+1}-1}{8} + \frac{3^m-1}{2} & \mbox{if } m \equiv 1 \pmod{4}.                 
               \end{array} 
\right.                
\end{eqnarray*}   
Develop 
tight lower bounds on the minimum weight of this code.   
\end{open}  

\begin{open} 
Determine the dimension and the generator polynomial of the ternary code $\C_s$ defined by 
the  APN function $f(x)=x^{e}$ over $\gf(3^m)$, where  
$$ 
e=\left(3^{(m+1)/4}-1\right)\left(3^{(m+1)/2}+1\right),  \ m \equiv 3 \pmod{4}.
$$  
Develop 
tight lower bounds on the minimum weight of this code.   
\end{open}

\section{A related construction of cyclic codes from highly nonlinear functions}\label{sec-newconstruct} 

Given any nonlinear function $f(x)$ on $\gf(r)$, we define its differential sequence 
$\seq^\infty$ by 
\begin{eqnarray}\label{eqn-difsequence}
\seq_i=\tr(f(\alpha^i+1)-f(\alpha^i)) 
\end{eqnarray}
for all $i \ge 0$, where $\alpha$ is a generator of $\gf(r)^*$ and $\tr(x)$ denotes 
the trace function from $\gf(r)$ to $\gf(q)$. We use $\C_{\seq}$ to denote the cyclic 
code defined by the sequence $\seq^\infty$. 

Polynomials over $\gf(r)$ of the form 
$$ 
\sum_{i,j} a_{i,j} x^{q^i + q^j}
$$
are called {\em Dembowski-Ostrom polynomials}, where $a_{i,j} \in \gf(r)$. 
If $f(x)$ is a Dembowski-Ostrom planar function, then $u(x)=f(x+1)-f(x)$ is an  
affine permutation of $\gf(r)$. In this case, one can prove that the code $\C_{\seq}$ 
has generator polynomial $(x-1)^\delta m_{\alpha^{-1}}(x)$ and parameters 
$[n, n-m-\delta, d]$, where $\delta \in \{0,1\}$ and 
\begin{eqnarray*} 
\left\{ \begin{array}{l} 
d=3 \mbox{ if } \delta=1, \\
d=2 \mbox{ if } \delta=0,  
 \end{array} 
 \right. 
 \end{eqnarray*} 
The code $\C_{\seq}$ is optimal in this case. 

When $f(x)=x^h$ is a monomial over $\gf(r)$, the relation between  codes $\C_{\seq}$ and $\C_s$ 
must be one of the following: 
\begin{itemize} 
\item The generator polynomial of $\C_s$ is equal to $m_{\alpha^{-h}}(x)$ times that of $\C_{\seq}$ 
          and the dimension of $\C_s$ is equal to that of $\C_{\seq}$ minus $m$. Hence, $\C_s$ is a 
          subcode of $\C_{\seq}$.    
\item The generator polynomial of $\C_\seq$ is equal to $m_{\alpha^{-h}}(x)$ times that of $\C_{s}$ 
          and the dimension of $\C_\seq$ is equal to that of $\C_{s}$ minus $m$. Hence, $\C_\seq$ is a 
          subcode of $\C_{s}$.                      
\end{itemize} 
 
All of the APN and planar functions could be plugged into this related construction above 
and the codes $\C_{\seq}$ are extremely good (many are optimal and almost optimal). 
Theorems about the codes $\C_s$ developed in previous sections can be modified into 
theorems about the codes $\C_{\seq}$. As an example, we will do this for the Welch APN 
function below. The rest of the modifications is left to the reader.

\begin{lemma}\label{lem-nWelch} 
Let $m =2t+1 \ge 7$. 
Let $\seq^{\infty}$ be the sequence of (\ref{eqn-difsequence}), where $f(x)=x^{2^{t}+3}$. 
Then the linear span $\ls_\seq$ of $\seq^{\infty}$ is equal to $4m+1$ and the minimal polynomial $\m_\seq(x)$ 
of  $\seq^{\infty}$ is given by 
\begin{eqnarray}\label{eqn-nWelch}
\m_\seq(x)=  (x-1) m_{\alpha^{-1}}(x) m_{\alpha^{-3}}(x) m_{\alpha^{-(2^t+1)}}(x)m_{\alpha^{-(2^t+2)}}(x)   
\end{eqnarray} 
where $m_{\alpha^{-j}}(x)$ is the minimal polynomial of $\alpha^{-j}$ over $\gf(2)$. 
\end{lemma} 

\begin{proof} 
The proof of Lemma \ref{lem-Welch} can be slightly modified into a proof of this lemma. 
The details are left to the reader.  
\end{proof} 

The following theorem provides information on $\C_{\seq}$.    

\begin{theorem} 
Let $m \ge 7$ be odd. 
The binary code $\C_{\seq}$ defined by the sequence of Lemma \ref{lem-nWelch} has parameters 
$[2^m-1, 2^{m}-2-4m, d]$ and  generator polynomial $\m_\seq(x)$ of (\ref{eqn-nWelch}), where $d \ge 6$.  
\end{theorem} 

\begin{proof} 
The proof of Theorem \ref{thm-Welch} can be slightly modified into a proof of this theorem. 
The details are left to the reader.  
\end{proof} 

\begin{example} 
Let $m=3$ and $\alpha$ be a generator of $\gf(2^m)^*$ with $\alpha^3 + \alpha + 1=0$. 
Then 
$\C_{\seq}$ is a $[7, 6, 2]$ optimal binary cyclic code with generator polynomial  
$ 
x + 1.    
$
\end{example}

\begin{example} 
Let $m=5$ and $\alpha$ be a generator of $\gf(2^m)^*$ with $\alpha^5 + \alpha^2 + 1=0$. 
Then 
$\C_{\seq}$ is a $[31, 20, 6]$ optimal binary cyclic code with generator polynomial 
$ 
x^{11} + x^9 + x^8 + x^7 + x^2 + 1. 
$ 
\end{example} 

\begin{example} 
Let $m=7$ and $\alpha$ be a generator of $\gf(2^m)^*$ with $\alpha^7 + \alpha + 1=0$. 
Then 
$\C_{\seq}$ is a $[127, 98, 8]$ binary cyclic code with generator polynomial 
\begin{eqnarray*} 
\m_\seq(x) &=& x^{37} + x^{36} + x^{35} + x^{34} + x^{33} + x^{28} + x^{26} + x^{24} + \\
& &  x^{22} + x^{21} + x^{17} + x^{13} + x^9 + x^8 + x^7 + x^5 + x^4 + 1. 
\end{eqnarray*}
This code is not an optimal linear code, but may be the best binary cyclic code of length 127 
and dimension 98. 
\end{example}

\section{Concluding remarks and summary}

In this paper, we studied the codes derived from a number of highly nonlinear functions (including 
planar functions and almost perfect nonlinear functions). Many of these codes obtained from 
these functions are optimal or almost optimal. The dimension of some of the codes is flexible.  
We determined the minimum weight for some classes of cyclic codes, and developed tight 
lower bounds for some other classes of cyclic codes. The main results of this paper showed 
that the two approaches of constructing cyclic codes with planar and APN functions are 
quite promising. While it is rare to see optimal cyclic codes constructed with tools  
in algebraic geometry and algebraic function fields, the simple constructions of cyclic 
codes with monomials and trinomials over $\gf(r)$ employed in this paper are very impressive 
in the sense that it has produced many optimal and almost optimal cyclic codes.   

Binary sequences defined by the inverse APN function and the first Niho APN function have large 
linear span. The sequences defined by the Kasami APN function have also large linear span when 
$h$ is close to $m/4$. These sequences have also reasonable autocorrelation property. They may be 
employed in certain stream ciphers as keystreams. For example, the sequence defined by the 
inverse APN function is used in the stream cipher in \cite{wenpeiding}.  So the contribution of this 
paper in cryptography is the computation of the linear spans of these sequences.

\end{document}